\documentclass[10pt,journal,compsoc]{IEEEtran}
\usepackage{amsmath,amsfonts,amsthm}
\usepackage[utf8]{inputenc}
\usepackage{algorithm}%,algorithmic}
\usepackage{algpseudocode}
\usepackage{booktabs}
\usepackage{bm}
\usepackage{comment}
\usepackage{chngpage}
\usepackage{enumerate}
\usepackage{paralist}
\usepackage{extramarks}
\usepackage{fancyhdr}
\usepackage{graphicx,float,wrapfig}
\usepackage{lastpage}
\usepackage{setspace}
\usepackage{soul,color}
\usepackage{xspace}
\usepackage{indentfirst}

\usepackage[hidelinks]{hyperref}
\ifCLASSOPTIONcompsoc
  \usepackage[nocompress]{cite}
\else
  \usepackage{cite}
\fi
\makeatletter
\DeclareRobustCommand\onedot{\futurelet\@let@token\@onedot}
\def\@onedot{\ifx\@let@token.\else.\null\fi\xspace}
\def\eg{\emph{e.g}\onedot} 
\def\ie{\emph{i.e}\onedot}

\def\GG{G_\text{IntDep}}
\def\dd{d_\text{dep}}
\makeatother

\theoremstyle{plain} %{remark} %plain
\newtheorem{thm}{Theorem}%[section]
\newtheorem{lem}[thm]{Lemma}

\newtheorem{defn}{Definition}%[section]

\newif\iftechrep \techrepfalse
\long\def\iftechreport#1#2{\iftechrep #1\else #2\fi}
\catcode`T=12 \catcode`R=12	% \jobname is catcode 12 (other)
\def\checkiftechreport#1{
\expandafter\iistechreport#1TR. \techreptrue\fi}
\def\iistechreport#1TR#2.{\def\tmp{#2}\ifx\tmp\empty\else}
\catcode`T=11 \catcode`R=11
\def\checkTR{\checkiftechreport{\jobname}}
\checkTR
\begin{document}
\title{Robustness of Interdependent Random Geometric Networks}
%\author{Jianan Zhang, Edmund Yeh, and Eytan Modiano
\author{Jianan~Zhang,~\IEEEmembership{Student~Member,~IEEE,}
        Edmund~Yeh,~\IEEEmembership{Senior~Member,~IEEE,}
        and~Eytan~Modiano,~\IEEEmembership{Fellow,~IEEE}
\thanks{Part of the material in this paper was presented at the 54th Annual Allerton Conference on Communication, Control, and Computing, 2016.

J. Zhang and E. Modiano are with the Laboratory for Information and Decision Systems, Massachusetts Institute of Technology. E. Yeh is with the Electrical and Computer Engineering Department, Northeastern University.

This work was supported by DTRA grants HDTRA1-14-1-0058, HDTRA1-13-1-0021, and NSF grant CMMI-1638234.}
}

\IEEEtitleabstractindextext{%
\begin{abstract}
We propose an interdependent random geometric
graph (RGG) model for interdependent networks. Based on this
model, we study the robustness of two interdependent spatially
embedded networks where interdependence exists between geographically
nearby nodes in the two networks. We study the
emergence of the giant mutual component in two interdependent
RGGs as node densities increase, and define the percolation threshold
as a pair of node densities above which the giant mutual
component first appears. In contrast to the case for a single RGG,
where the
percolation threshold is a unique scalar for a given connection distance, for two interdependent RGGs,
multiple pairs of percolation thresholds may exist, given that a smaller
node density in one RGG may increase the minimum node density
in the other RGG in order for a giant mutual component to
exist.  We derive analytical upper bounds on the percolation
thresholds of two interdependent RGGs by discretization, and
obtain $99\%$ confidence intervals for the percolation thresholds
by simulation. Based on these results, we derive conditions for
the interdependent RGGs to be robust under random failures
and geographical attacks.
\end{abstract}

% Note that keywords are not normally used for peerreview papers.
}

% make the title area
\maketitle

% To allow for easy dual compilation without having to reenter the
% abstract/keywords data, the \IEEEtitleabstractindextext text will
% not be used in maketitle, but will appear (i.e., to be "transported")
% here as \IEEEdisplaynontitleabstractindextext when the compsoc
% or transmag modes are not selected <OR> if conference mode is selected
% - because all conference papers position the abstract like regular
% papers do.
\IEEEdisplaynontitleabstractindextext
% \IEEEdisplaynontitleabstractindextext has no effect when using
% compsoc or transmag under a non-conference mode.

% For peer review papers, you can put extra information on the cover
% page as needed:
% \ifCLASSOPTIONpeerreview
% \begin{center} \bfseries EDICS Category: 3-BBND \end{center}
% \fi
%
% For peerreview papers, this IEEEtran command inserts a page break and
% creates the second title. It will be ignored for other modes.
\IEEEpeerreviewmaketitle

\IEEEraisesectionheading{\section{Introduction}\label{sec:introduction}}
Cyber-physical systems such as smart power grids and smart transportation networks are being deployed towards the design of smart cities. The integration of communication networks and physical networks facilitates network operation and control. In these integrated networks, one network depends on another for information, power, or other supplies in order to properly operate, leading to interdependence. For example, in smart grids, communication networks rely on the electric power from power grids, and simultaneously control power generators \cite{rosato2008modelling, parandehgheibi2015modeling}. Failures in one network may cascade to another network, which potentially make the interdependent networks vulnerable. %Real world examples of electrical blackouts suggest the vulnerability of some interdependent networks.

Cascading failures in interdependent networks have been extensively studied in the statistical physics literature since the seminal work in \cite{buldyrev2010catastrophic}, where each of the two interdependent networks is modeled as a random graph. A node is functional if both itself and its interdependent node are in the giant components of their respective random graphs. After initial node failures in the first graph, their interdependent nodes in the second graph fail. Thus, a connected component in the second graph may become disconnected, and the failures of the disconnected nodes cascade back to (their interdependent) nodes in the first graph. As a result of the cascading failures, removing a small fraction of nodes in the first random graph destroys the giant components of both graphs.

To model spatially embedded networks, an interdependent lattice model was studied in \cite{bashan2013extreme}. Under this model, geographical attacks may cause significantly more severe cascading failures than random attacks. Removing nodes in a finite region (\ie, a zero fraction of nodes) may destroy the infinite clusters in both lattices \cite{berezin2015localized}.

If every node in one network is interdependent with multiple nodes in the other network, and a node is content to have at least one interdependent node, failures are less likely to cascade \cite{PhysRevEShao, yaugan2012optimal}. Although the one-to-multiple interdependence exists in real-world spatially embedded interdependent networks (\eg, a control center can be supported by the electric power generated by more than one power generator), it has not been previously studied using spatial graph models.

We use a random geometric graph (RGG) to model each of the two interdependent networks. \textcolor{black}{RGG has been widely used to model communication networks \cite{franceschetti2008random}. For example, in a wireless network where the communication distance is limited by the signal to noise ratio requirement, under fixed transmission power, two users can communicate if and only if they are within a given distance. Percolation theory for RGG has been applied to study information flow in wireless networks and the robustness of networks under failures \cite{Tse,kong2010}. In this paper, we extend percolation theory to interdependent RGGs.} %RGG is a more realistic model for spatially embedded networks, such as ad-hoc networks, compared with the square lattice. Moreover, t

The two RGGs representing two interdependent networks are allowed to have different connection distances and node densities, which can represent two networks that have different average link lengths and scales. These network properties were not captured by the lattice model in the previous literature. Moreover, the interdependent RGG model is able to capture the one-to-multiple interdependence in spatially embedded networks, and provides a more versatile framework for studying interdependent networks.

Robustness is a key design objective for interdependent networks. We study the conditions under which a positive fraction of nodes are functional in interdependent RGGs as the number of nodes approaches infinity. In this case, the interdependent RGGs {\it percolate}. Consistent with previous research \cite{buldyrev2010catastrophic, bashan2013extreme, PhysRevEShao}, the robustness of interdependent RGGs under failures is measured by whether percolation exists after failures. To the best of our knowledge, our paper is the first to study the percolation of interdependent spatial network models using a mathematically rigorous approach. %This metric is in consistent with the previous research in interdependent networks. While it is important to understand cascading failures, it is equally important to obtain conditions under which the interdependent RGGs are robust in order to design interdependent networks, which is the focus of this paper.

The main contributions of this paper are as follows.
\begin{enumerate}
\item We propose an interdependent RGG model for two interdependent networks, which captures the differences in the scales of the two networks as well as the one-to-multiple interdependence in spatially embedded networks.
\item We derive the first analytical upper bounds on the percolation thresholds of the interdependent RGGs, above which a positive fraction of nodes are functional.
\item We obtain $99\%$ confidence intervals for the percolation thresholds, by mapping the percolation of interdependent RGGs to the percolation of a square lattice where the probability that a bond in the square lattice is open is evaluated by simulation.
\item We characterize sufficient conditions for the interdependent RGGs to percolate under random failures and geographical attacks. In particular, if the node densities are above any upper bound on the percolation threshold obtained in this paper, the interdependent RGGs remain percolated after a geographical attack. This is in contrast with the cascading failures after a geographical attack, observed in the interdependent lattice model with one-to-one interdependence \cite{berezin2015localized}.
\item We extend our techniques to study models with more general interdependence requirement (\eg, a node in one network requires more than one supply node from the other network).
\end{enumerate}

The rest of the paper is organized as follows. We state the model and preliminaries in Section \ref{sc:model}. We derive analytical upper bounds on percolation thresholds in Section \ref{sc:upperbound}, and obtain confidence intervals for percolation thresholds in Section \ref{sc:interval}. In Section \ref{sc:robustness}, we study the robustness of interdependent RGGs under random failures and geographical attacks. In Section \ref{sc:extension}, we extend the techniques to study graphs with more general interdependence. Section \ref{sc:conclusion} concludes the paper.

\section{Model}
\label{sc:model}
\subsection{Preliminaries on RGG and percolation}
An RGG in a two-dimensional square consists of nodes generated by a Poisson point process and links connecting nodes within a given connection distance \cite{penrose2003}. Let $G(\lambda, d, a^2)$ denote an RGG with node density $\lambda$ and connection distance $d$ in an $a \times a$ square. The studies on RGG focus on the regime where the expected number of nodes $n = \lambda a^2$ is large. We first present some preliminaries which are useful for developing our model. The {\it giant component} of an RGG is a connected component that contains $\Theta(n)$ nodes. A node belongs to the giant component with a positive probability $\Theta(n)/n$ if the giant component exists. For a given connection distance, the {\it percolation threshold} is a node density above which a node belongs to the giant component with a positive probability (\ie, a giant component exists) and below which the probability is zero (\ie, no giant component exists). By scaling, if the percolation threshold is $\lambda^*$ under connection distance $d$, then the percolation threshold is $\lambda^* c^2$ under connection distance $d/c$. Therefore, without loss of generality, in this paper, we study the percolation thresholds represented by node densities, for given connection distances.

The RGG is closely related to the \emph{Poisson boolean model} \cite{meester1996continuum}, where nodes are generated by a Poisson point process on an {\it infinite plane}. Let $G(\lambda, d)$ denote a Poisson boolean model with node density $\lambda$ and connection distance $d$. The difference between $G(\lambda, d)$ and $G(\lambda, d, a^2)$ is that the number of nodes in $G(\lambda, d)$ is infinite while the expected number of nodes in $G(\lambda, d, a^2)$ is large but finite. The Poisson boolean model can be viewed as a limit of the RGG as the number of nodes approaches infinity. The percolation threshold of $G(\lambda, d)$ under a given $d$ is defined as the node density above which a node belongs to the {\it infinite component} with a positive probability and below which the probability is zero. It has been shown that a node belongs to the infinite component with a positive probability if and only if an infinite component exists, and thus the percolation of $G(\lambda, d)$ can be equivalently defined as the existence of the infinite component \cite{meester1996continuum}. Moreover, the percolation threshold of $G(\lambda, d)$ is identical with the percolation threshold of $G(\lambda, d, a^2)$ \cite{penrose2003, balister2008percolation}.

%Some properties of the Poisson boolean model are helpful to understand the results in this paper. %An equivalent definition for the connections between nodes is to connect two nodes if the disks of radius $r = d/2$ centered at the two nodes overlap. Clearly two nodes are within distance $d$ if and only if the two disks overlap. Let the {\it covered region} be the union of the disks, and the {\it vacant region} be the complement of the disk regions. Below the percolation threshold, there exists a vacant region that has infinite size while all the covered regions have finite sizes. Above the percolation threshold, a covered region of infinite size contains an infinite component while all the vacant regions have finite sizes.
\subsection{Interdependent RGGs}
Two interdependent networks are modeled by two RGGs $G_1(\lambda_1, d_1, a^2)$ and $G_2(\lambda_2, d_2, a^2)$ on the \emph{same} $a \times a$ square. A node in one graph is interdependent with {\it all} the nodes in the other graph within the \emph{interdependent distance} $d_{\text{dep}}$. See Fig. \ref{fig:IntDepRGG} for an illustration.
Nodes in one graph are {\it supply nodes} for nodes in the other graph within $\dd$. The physical interpretation of supply can be either electric power or information that is essential for proper operation. A node can receive supply from nearby nodes within the interdependent distance. Larger interdependent distance leads to more robust interdependent networks. The geographical nature of interdependence is observed in physical networks \cite{rosato2008modelling, bashan2013extreme}.  %, and each node can receive supply from any of its supply nodes.
\begin{figure}[h]
\begin{centering}
\leavevmode\includegraphics[width=0.7\linewidth]{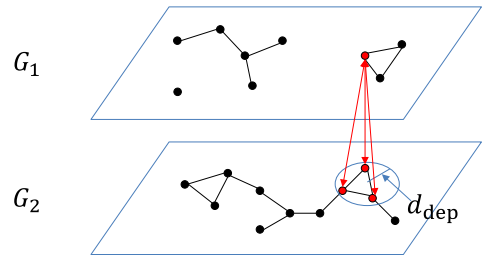}
\caption{Two interdependent RGGs with interdependent distance $\dd$.}
\label{fig:IntDepRGG}
\end{centering}
\end{figure}

Most analysis in this paper is given in the context of two interdependent Poisson boolean models $\GG = (G_1(\lambda_1, d_1), G_2(\lambda_2, d_2), \dd)$, which is the limit of two interdependent RGGs as the numbers of nodes in both graphs approach infinity.

We define a mutual component and an infinite mutual component in $\GG$, in the same way as one defines a connected component and an infinite component in $G(\lambda,d)$.

\begin{defn}
Let $V^0_i$ denote nodes in a connected component in $G_i(\lambda_i, d_i)$, $\forall i \in \{1,2\}$. If each node in $V_i \subseteq V^0_i$ has at least one supply node in $V_j \subseteq V^0_j$ within $d_{\text{dep}}$, $\forall i, j \in \{1,2\}, i \neq j$, then nodes $V_1$ and $V_2$ form a \emph{mutual component} of $\GG$.

If, in addition, $V_i$ contains an infinite number of nodes, $\forall i\in \{1,2\}$, then $V_1$ and $V_2$ form an \emph{infinite mutual component}.
\end{defn}

A mutual component can be viewed as an autonomous system in the sense that nodes in a mutual component have supply nodes in the same mutual component, and in each graph, nodes that belong to a mutual component are connected regardless of the existence of nodes outside the mutual component. \textcolor{black}{Note that a node can receive supply from any of its supply nodes in the same mutual component, and thus is content if it has at least one supply node}. Nodes in an infinite mutual component are \emph{functional}, since they constitute two large connected interdependent networks and can perform a given network function (\eg, data communication or power transmission to a large number of clients). This definition of functional is consistent with previous research on interdependent networks based on random graph models~\cite{buldyrev2010catastrophic}.

%Percolation defined as the existence of a giant mutual component is consistent with the definition in previous interdependent networks literature based on random graph models.
For a fixed $\dd$, if an infinite mutual component exists in $\GG = (G_1(\lambda_1, d_1), G_2(\lambda_2, d_2), \dd)$, then an infinite mutual component exists in $\GG' = (G_1(\lambda'_1, d_1), G_2(\lambda_2, d_2), \dd)$, where $\lambda'_1 > \lambda_1$. This can be explained by coupling $G'_1$ with $G_1$ as follows. By removing each node in $G'_1$ independently with probability $1 - \lambda_1/\lambda'_1$, the density of the remaining nodes in $G'_1$ is $\lambda_1$, and an infinite mutual component exists in the interdependent graphs that consist of $G_2$ and the graph formed by the remaining nodes in $G'_1$. Since adding nodes to a graph does not disconnect any mutual component, an infinite mutual component exists in $\GG' = (G_1(\lambda'_1, d_1), G_2(\lambda_2, d_2), \dd)$. By the same analysis, an infinite mutual component also exists in $\GG'' = (G_1(\lambda_1, d_1), G_2(\lambda'_2, d_2), \dd)$, if $\lambda'_2 > \lambda_2$.

We define a percolation threshold of $\GG$ as follows.
\begin{defn}
A pair of node densities $(\lambda_1^*, \lambda_2^*)$ is a \emph{percolation threshold} of $\GG$, given connection distances $d_1, d_2$ and the interdependent distance $\dd$, if an infinite mutual component exists in $\GG$ for $\lambda_1 > \lambda_1^*$ and $\lambda_2 > \lambda_2^*$, and no infinite mutual component exists otherwise.
\end{defn}

For fixed $d_1$, $d_2$ and $\dd$, there may exist multiple percolation thresholds. We show that, in most cases, the larger the node density is in one graph, the smaller the required node density is in the other graph in order for the infinite mutual component to exist. This is in contrast with the situation for a single graph $G(\lambda,d)$ where there is a unique percolation threshold $\lambda^*$ for a fixed $d$.

There is a non-trivial phase transition in $\GG$. If $\lambda_i$ is smaller than the percolation threshold of a single graph $G_i(\lambda_i, d_i)$, there is no infinite component in $G_i(\lambda_i, d_i)$, and therefore there is no infinite mutual component in $\GG$. Thus, $\lambda_i^* > 0$, $\forall i \in \{1,2\}$. As we will see in the next section, there exist percolation thresholds $\lambda_i^* < \infty$, $\forall i \in \{1,2\}$, which concludes the non-trivial phase transition.

Given that the conditions for the percolation of a random geometric graph $G_i(\lambda_i, d_i, a^2)$ and a Poisson boolean model $G_i(\lambda_i, d_i)$ are the same, the above definitions can be naturally extended to interdependent RGGs. Consider nodes $V_1 \subseteq G_1(\lambda_1, d_1, a^2)$ and $V_2 \subseteq G_2(\lambda_2, d_2, a^2)$ that form a mutual component. If $V_i$ contains $\Theta(n_i)$ nodes, where $n_i = \lambda_i a^2$, $\forall i\in \{1,2\}$, then $V_1$ and $V_2$ form a \emph{giant mutual component} in interdependent RGGs. The percolation of interdependent RGGs is defined as the existence of a giant mutual component. In the rest of the paper, we sometimes use $G_i$ to denote both $G_i(\lambda_i, d_i, a^2)$ and $G_i(\lambda_i, d_i)$. The model that it refers to will be clear from the context.

\color{black}
\subsection{Related work}
In the interdependent networks literature, the model which is closest to ours is the interdependent lattice model, first proposed in \cite{li2012cascading} and further studied in \cite{bashan2013extreme, berezin2015localized}. In the lattice model, nodes in a network are represented by the open {\it sites} (nodes) of a square lattice, where every site is open independently with probability $p$. Network links are represented by the {\it bonds} (edges) between adjacent open sites. Every node in one lattice is interdependent with {\it one} randomly chosen node within distance $r_d$ in the other lattice. The distance $r_d$ indicates the geographical proximity of the interdependence. The percolation threshold of the interdependent lattice model is characterized as a function of $r_d$, assuming the same $p$ in both lattices \cite{li2012cascading}. Percolation of the model where some nodes do not need to have supply nodes was studied in \cite{bashan2013extreme}. The analysis relies on quantities estimated by simulation and extrapolation, such as the fraction of nodes in the infinite component of a lattice for any fixed $p$, which cannot be computed rigorously. In contrast, we study the percolation of the interdependent RGG model using a mathematically rigorous approach.

The percolation of a single RGG (or a Poisson boolean model) has been studied in the previous literature \cite{hall1985continuum, meester1996continuum, balister2005}. The techniques employed therein involves inferring the percolation of the continuous model from the percolation of a discrete lattice model. The key is obtaining a lattice whose percolation condition is known and is related to the percolation of the original model, by discretization. The study of the percolation conditions of discrete lattice models can be found in \cite{Grimmett1999, bollobas2006percolation}. %Bounds on the percolation threshold of the Poisson boolean model can be obtained by the discretization approach.
We extend the previous techniques to discretize $\GG$, and obtain bounds on the percolation thresholds.

%our results in Section \ref{sc:upperbound} that calculate upper bounds on percolation thresholds are purely analytical. Although the results in Section \ref{sc:interval} involve simulations, they are used to estimated the probability of an event that only depends on a finite region, instead of estimating the .
\section{Analytical upper bounds on percolation thresholds}
\label{sc:upperbound}
\textcolor{black}{In this section, we study sufficient conditions for the percolation of $\GG$. We provide closed-form formulas for $(\lambda_1, \lambda_2)$, which depend on $d_1,d_2,\dd$, such that there exists an infinite mutual component in $\GG = (G_1(\lambda_1,d_1),G_2(\lambda_2,d_2),\dd)$. The formulas provide guidelines for node densities in deploying physical interdependent networks, in order for a large number of nodes to be connected.}

In $\GG$, nodes in the infinite mutual component are viewed as functional while all the other nodes are not. Thus, a node is functional only if it is in the infinite component of its own graph, and it depends on at least one node in the infinite component of the other graph. For any node $b_1$ in $G_1$, although the number of nodes in $G_2$ within the interdependent distance from $b_1$ follows a Poisson distribution, the number of functional nodes is hard to calculate, since the probability that a node in $G_2$ is in the infinite component is unknown. Moreover, the nodes in the infinite component of $G_2$ are clustered, and thus the thinning of the nodes in $G_1$ due to a lack of supply nodes in $G_2$ is inhomogeneous. To overcome these difficulties, we consider the percolation of two graphs jointly, instead of studying the percolation of one graph with reduced node density due to a lack of supply nodes.

We now give an overview of our approach. We develop mapping techniques (discretizations) to characterize the percolation of $\GG$ by the percolation of a discrete model. Mappings from a model whose percolation threshold is unknown to a model with known percolation threshold are commonly employed in the study of continuum percolation. For example, one can study the percolation threshold of the Poisson boolean model $G(\lambda, d)$ by mapping it to a triangle lattice and relating the state of a site in the triangle lattice to the point process of $G(\lambda,d)$. By the mapping, the percolation of the triangle lattice implies the percolation of $G(\lambda, d)$. Consequently, an upper bound on the percolation threshold of $G(\lambda, d)$ is given by $\lambda$ for which the triangle lattice percolates, a known quantity \cite{hall1985continuum, meester1996continuum}. In general, more than one mapping can be applied, and the key is to find a mapping that gives a good (smaller) upper bound. Following this idea, we propose different mappings that fit different conditions to obtain upper bounds on the percolation thresholds of $\GG$.

\iftechreport{
In the rest of this section, we first study an example, in which the connection distances of the two graphs are the same, to understand the tradeoff between the two node densities in order for $\GG$ to percolate.
We then develop two upper bounds on the percolation thresholds. The first bound is tighter when the ratio of the two connection distances is small, and is obtained by mapping $\GG$ to a square lattice with independent bond open probabilities. The second bound is tighter when the ratio of the two connection distances is large, and is obtained by mapping $\GG$ to a square lattice with correlated bond open probabilities.
}{}
%To obtain upper bounds on percolation thresholds in the general case, we map $\GG$ to a square lattice with independent bond open probabilities if the difference between the two connection distances is small, and map $\GG$ to a square lattice with correlated bond open probabilities if the difference between the two connection distances is large.
%\begin{comment}
\iftechreport{
\subsection{A motivating example}
\label{sc:example}
To see the impact of varying the node density in one graph on the minimum node density in the other graph in order for $\GG$ to percolate, consider an example where $d_1 = d_2 = 2d_{\text{dep}}$. We apply a mapping similar to what is used to obtain an upper bound on the percolation threshold of $G(\lambda,d)$ in \cite{hall1985continuum}, to obtain upper bounds on the percolation thresholds of $\GG$.

Consider a triangle lattice where each site is surrounded by a cell. The lattice bond length is determined such that any two points in adjacent cells have distance smaller than $2r$, where $2r = d_1$. The boundary of the cell consists of arcs of radius $r$ centered at the middle of the bonds in the triangle lattice. See Fig.~\ref{lattice} for an illustration. The area of the cell is $A=0.8227 r^2$. A site in the triangle lattice is either {\it open} or {\it closed}. If the probability that a site is open is strictly larger than $1/2$, open sites form an infinite component, and the triangle lattice percolates \cite{hall1985continuum}.
\begin{figure}[h]
\begin{centering}
\leavevmode\includegraphics[width=0.55\linewidth]{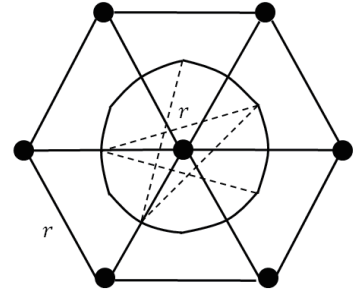}
\caption{A cell that contains a site in a triangle lattice.}
\label{lattice}
\end{centering}
\end{figure}

To study the percolation of $\GG$, we declare a site in the triangle lattice to be open if there is at least one node in its cell from $G_1$ and at least one node in its cell from $G_2$. If the triangle lattice percolates, then $\GG$ also percolates. To see this, consider two adjacent open sites in the triangle lattice. Nodes from $G_i$ in the two adjacent cells that contain the two open sites are connected, because they are within distance $d_i = 2r$ ($\forall i \in \{1,2\}$). If the open sites in the triangle lattice form an infinite component, then nodes from $G_i$ in the corresponding cells form an infinite component $V_i$ ($\forall i \in \{1,2\}$). Moreover, given that any pair of nodes in a cell are within distance $r \leq \dd$, each node in $V_i$ has at least one supply node in $V_j$ within the same cell ($\forall i,j \in \{1,2\}, i \neq j$).

Since $1 - e^{-\lambda_i A}$ is the probability that there is at least one node in the cell from $G_i$ and the point processes in $G_1$ and $G_2$ are independent, an upper bound on the percolation thresholds of $\GG$ is given by $(\lambda_1, \lambda_2)$ satisfying
$$ (1 - e^{-\lambda_1 A})(1 - e^{-\lambda_2 A}) = 1/2. $$

%The upper bounds are depicted in Fig. \ref{percolationexample}.
If $\lambda_i$ is large, the percolation threshold $\lambda^*_j$ approaches the threshold of a single graph $G_j$. Intuitively, if $\lambda_j$ is above the percolation threshold of $G_j$, disks of radius $d_j / 2$ centered at nodes in $G_j$ form a connected infinite-size region. Since $\lambda_i$ is large, nodes in $G_i$ in this region are connected and form an infinite component. Moreover, since $d_\text{dep} = d_j / 2$, all the nodes in this region have supply nodes, and they form an infinite mutual component.

%If $\lambda_1 = \lambda_2$, an upper bound is $\lambda_1 = \lambda_2 = 5.970/d^2$, which is larger than $3.370/d^2$. In general, higher node densities than the percolation threshold of each individual graph are required in order for the interdependent graphs to percolate, except when $\dd$ or the node density of the other graph is so large that every node in one graph has many supply nodes in the other graph.

The above upper bounds on percolation thresholds are still valid if $d_\text{dep} > d_i / 2$, because each node can depend on a larger set of nodes by increasing $d_\text{dep}$ and it is easier for $\GG$ to percolate under the same node densities and connection distances. However, if $d_\text{dep} < d_i / 2$, the bond length of the triangle lattice should be adjusted to $r = d_\text{dep}$ in order for any pair of nodes in a cell to be within $d_\text{dep}$. The percolation threshold curve $(\lambda_1, \lambda_2)$ would shift upward. Intuitively, if $d_\text{dep}$ decreases, the node density in one network should increase to provide enough supply for the other network.
}{}
\color{black}
\subsection{Small ratio $d_2/d_1$}
Given $\GG = (G_1(\lambda_1,d_1), G_2(\lambda_2,d_2), \dd)$, without loss of generality we assume that $d_1 \leq d_2$. Moreover, we assume that $d_\text{dep} \geq \text{max}(d_1/2, d_2/2) = d_2/2$ (see the remark at the end of the section for comments on this assumption).
Let $c = \lfloor d_2/d_1 \rfloor = \max\{c: d_2/d_1 \geq c, c \in \mathbb{N}\}$. For small $c$, we study the percolation of $\GG$ by mapping it to an independent bond percolation of a square lattice, and prove the following result.

%We map the interdependent RGGs to a square lattice to study the percolation properties.
\begin{thm}
\label{th:uppersmall}
If $(\lambda_1,\lambda_2)$ satisfies
$$ (1 - e^{- \lambda_1 d_1^2 / 8})^c (1 - e^{- \lambda_2 c^2 d_1^2 / 8}) > 1/2,$$
%$$[1 - \exp(-\frac{d_1^2}{8} \lambda_1)]^c [1 - \exp(-\frac{c^2d_1^2}{8} \lambda_2 )] > 1/2,$$
then $\GG = (G_1(\lambda_1, d_1), G_2(\lambda_2, d_2), \dd)$ percolates, where $c = \lfloor d_2/d_1 \rfloor$, $d_1 \leq d_2$, and $\dd \geq d_2/2$.
\end{thm}

\textcolor{black}{Theorem \ref{th:uppersmall} provides a sufficient condition for the percolation of $\GG$. For node densities that satisfy the inequality, an infinite mutual component exists in $\GG$. For the deployment of interdependent networks, if the node densities in the two networks are sufficiently large (characterized by Theorem \ref{th:uppersmall}), then a large number of nodes in the interdependent networks are functional.}

\begin{proof}[Proof of Theorem \ref{th:uppersmall}]
%The proof is based on a mapping from the percolation of $\GG$ to the independent bond percolation of a square lattice.
We first construct a square lattice as follows. Partition the plane into small squares of side length $s = d_1/2\sqrt{2}$. A large square consists of $c \times c$ small squares and has side length $cs$. The {\it diagonals} of the large squares form the bonds of a square lattice $L$, illustrated by the thick line segments in Fig.~\ref{squarelattice}.

The state of a bond in $L$ is determined by the point process of $\GG$ in the large square that contains the bond.
A bond $(v_1,v_2)$ is open if the following conditions are both satisfied.
\begin{enumerate}
\item There is at least one node from $G_1$ in each of the two small squares that contain the ends ($v_1$ and $v_2$) of the bond, and they are connected through nodes from $G_1$, \emph{all within the large square of side length~$cs$}.
\item There is at least one node from $G_2$ in the large square that contains the bond.
\end{enumerate}
 %We aim to map the percolation of $\GG$ to the bond percolation of $L$ to obtain upper bounds on the percolation thresholds of $\GG$.
\begin{figure}[h]
\begin{centering}
\leavevmode\includegraphics[width=0.9\linewidth]{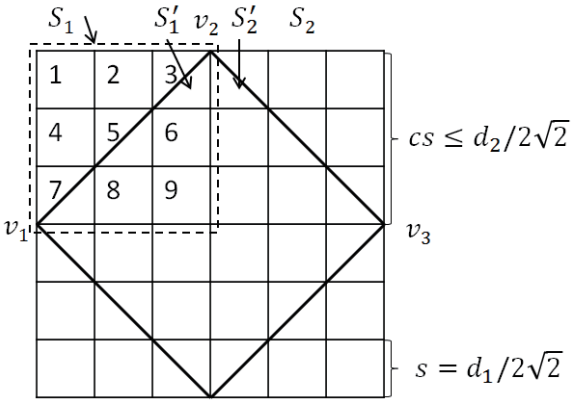}
\caption{Mapping to a square lattice for $c = 3$.}
\label{squarelattice}
\end{centering}
\end{figure}

%Two squares are adjacent if they share a common bond or site. If two nodes from $G_1$ are in adjacent small squares of side length $s$, they are within distance $2\sqrt 2 s = d_1$, and are connected by a link in $G_1$. If two nodes from $G_2$ are in adjacent large squares of side length $cs$, they are within distance $2\sqrt 2 cs \leq d_2$ and are connected by a link in $G_2$. Let $v_1$ and $v_2$ be the two ends of the diagonal of the large square.

%$$ p_1 = 1 - \exp(-\frac{c^2 d_1^2}{8} \lambda_2).$$

The first condition is satisfied if there exists a sequence of adjacent small squares, each of which contains at least one node in $G_1$, from the small square that contains $v_1$ to the small square that contains $v_2$. (Each small square is \emph{adjacent} to its eight immediate neighbors.) In the example of Fig. \ref{squarelattice}, these sequences include 3-5-7, 3-2-4-7, and 3-6-8-7.

To obtain a closed-form formula, instead of computing the exact probability, we compute a lower bound on the probability that the first condition is satisfied. The probability is lower bounded by the probability that the $c$ small squares that intersect the bond each contain at least one node from $G_1$, given by
$$ p_1 \geq (1 - e^{- \lambda_1 d_1^2 / 8})^c. $$
%$$ p_2 \geq (1 - \exp(-\frac{d_1^2}{8} \lambda_1))^c. $$

The probability that the second condition is satisfied is
$$ p_2 = 1 - e^{- \lambda_2 c^2 d_1^2 / 8}. $$

Given that the two Poisson point processes in $G_1$ and $G_2$ are independent, the probability that a bond is open is $p_1p_2$.

It remains to prove that the percolation of $L$ implies the percolation of $\GG$. Consider two adjacent open bonds $(v_1,v_2),(v_2,v_3)$ in $L$. Let $S_1$ and $S_2$ denote the two adjacent large squares of side length $cs$ that contain the two open bonds. Let $S'_1$ and $S'_2$ denote two adjacent small squares of side length $s$ that contains $v_2$, within $S_1$ and $S_2$, respectively. See Fig.~\ref{squarelattice} for an illustration. Since $(v_1,v_2),(v_2,v_3)$ are open, under the second condition, nodes of $G_2$ exist in $S_1$ and $S_2$ and they are connected, because they are within distance $2 \sqrt 2 cs \leq d_2$. Under the first condition, nodes of $G_1$ form a connected path from the small square (within $S_1$, marked as 7 in Fig.~\ref{squarelattice}) containing $v_1$ to $S'_1$, and another path from the small square (within $S_2$) containing $v_3$ to $S'_2$. Moreover, the two paths are joined, because any pair of nodes in $S'_1$ and $S'_2$ are within distance $2\sqrt 2 s = d_1$. Given that any pair of nodes within a large square have distance at most $\sqrt2 c s \leq d_2/2 \leq d_\text{dep}$, all the nodes have at least one supply node inside the large square that contains an open bond. To conclude, if the open bonds in $L$ form an infinite component, then the nodes in $\GG$ form an infinite mutual component.

The event that a bond is open depends on the point processes in the large square that contains the bond, and is independent of whether any other bonds are open. As long as the probability that a bond is open, $p_1p_2$, is larger than $1/2$, which is the threshold for independent bond percolation in a square lattice \cite{bollobas2006percolation}, $\GG$ percolates. %To conclude, we have proved the following result.
\end{proof}
%The following theorem provides upper bounds on the percolation thresholds of $\GG$.

%Theorem \ref{th:uppersmall} provides an upper bound on the percolation threshold for general values of $d_1$ and $d_2$.
The bound can be made tighter for any given $c = \lfloor d_2/d_1 \rfloor$, by computing more precisely the probability that the first condition is satisfied. We provide an example to illustrate the computation of an improved upper bound.

{\it Example:} Consider an example where $d_1 = 1, d_2 = 2d_\text{dep} = 3$. The probability that there is at least one node from $G_2$ in the large square of side length $3/2\sqrt2$ is
$ p_2 = 1 - e^{-9\lambda_2/8}.$

The probability that a small square of side length $1/2\sqrt2$ contains at least one node from $G_1$ is $ p_s = 1 - e^{- \lambda_1/8}.$
The probability that the first condition is satisfied is
\begin{equation}
p_1 = p_s^3 + (1-p_s)p_s^4 + (1-p_s)p_s^4 - (1 - p_s) p_s^6,
\label{eq:newp}
\end{equation}
obtained by considering all the sequences of adjacent small squares. For node densities $(\lambda_1,\lambda_2)$ that satisfy $p_1p_2>1/2$, $\GG$ percolates. Since $p_1$ computed by Eq. (\ref{eq:newp}) is larger than $p_s^3$ for any fixed $p_s$, the bound on $\lambda_2$ is smaller for any fixed $\lambda_1$.

\subsection{Large ratio $d_2/d_1$}
\label{sc:onedeplarge}
%If $c$ is large and the node density of $G_1$ is above the percolation threshold, the probability that there exists a path in $G_1$ close to the diagonal of the $cs \times cs$ square shown in Fig. \ref{squarelattice} approaches one for $s = d_1/2\sqrt{2}$. However, in the above mapping, this probability is limited by $(1 - e^{- \lambda_1/2})^2$, given that each of the small squares that contain the ends of the diagonal each must have at least one node from $G_1$.
In the mapping from $\GG$ to the square lattice $L$, the condition for a bond to be open becomes overly restrictive as $d_2/d_1$ increases. A path crossing the two large squares that contain two adjacent bonds does not have to cross the small squares that contain the common end of the two bonds.
\begin{comment}
See Fig.~\ref{largediff} for an example. The (light) blue path consists of nodes in $G_1$ that satisfy the condition. However, even without the blue path, another (dark) red path may exist, which does not satisfy the condition but is still a connected path in $G_1$ near the two bonds. With some extra efforts, the percolation of $\GG$ can be proved.
\begin{figure}[h]
\begin{centering}
\leavevmode\includegraphics[width=0.6\linewidth]{largediff.eps}
\caption{Many paths near the diagonals may exist in $G_1$.}
\label{largediff}
\end{centering}
\end{figure}
\end{comment}
\textcolor{black}{In the following theorem, we give another upper bound on the percolation threshold of $\GG$. This result provides an alternative sufficient condition for the existence of an infinite mutual component in $\GG$. This upper bound is tighter than the bound in Theorem \ref{th:uppersmall} for larger values of $d_2/d_1$.}

\begin{thm}
If $(\lambda_1, \lambda_2)$ satisfies
{\small $$ \Big[1 - \frac{4}{3}(m+1) e^{m \log 3(1-p)}\Big]\Big[1 - \frac{4}{3}(2m+1) e^{m \log 3(1-p)}\Big] p' > 0.8639,$$}
then $\GG = (G_1(\lambda_1, d_1), G_2(\lambda_2, d_2), \dd)$ percolates, where $p = 1 - e^{-\lambda_1 d_1^2 / 8}$, $p' = 1 - e^{-2D^2 \lambda_2}$, $D = \min (d_2 / \sqrt{10}, d_\text{dep} / \sqrt 5), m = \lfloor 2D/d_1 \rfloor$, $d_1 \leq d_2$, and $\dd \geq d_2/2$.
\label{th:upperlarge}
\end{thm}

This upper bound is obtained by mapping $\GG$ to a dependent bond percolation model $L_D$. The mapping from the Poisson boolean model $G(\lambda, d)$ to $L_D$ was first proposed in \cite{balister2005} to study the percolation threshold of $G(\lambda, d)$, and later applied to the study of a random geometric graph under non-uniform node removals \cite{kong2010}. We briefly describe the method in the previous literature that uses $L_D$ to study the percolation of $G(\lambda, d)$, and then prove Theorem \ref{th:upperlarge} based on a similar method.
\subsubsection{1-dependent bond percolation model $L_D$}
\label{sc:large}
In the standard bond percolation model on a square lattice $L$, the event that a bond is open is independent of the event that any other bond is open. If in a square lattice $L_D$, the event that a bond is open may depend on the event that its adjacent bond is open, but is independent of the event that any non-adjacent bond is open, then $L_D$ is a {\it 1-dependent bond percolation model} on a square lattice. With the additional restriction that each bond is open with an identical probability, an upper bound on the percolation threshold of $L_D$ is 0.8639~\cite{balister2005}.

The 1-dependent bond percolation model $L_D$ can be used to study the percolation of $G'$ where the points are generated by homogeneous Poisson point processes. %Both $G(\lambda,d)$ and $\GG$ are examples of~$G'$.
To construct a mapping from $G'$ to $L_D$, consider two adjacent $D \times D$ squares $S_1$ and $S_2$ and let $R$ be the rectangle formed by the two squares. A bond $(v_1, v_2)$ that connects the centers of $S_1$ and $S_2$ is associated with $R$. Figure \ref{onedependent} illustrates the square lattice formed by the bonds, represented by thick line segments.
\begin{figure}[h]
\begin{centering}
\leavevmode\includegraphics[width=0.85\linewidth]{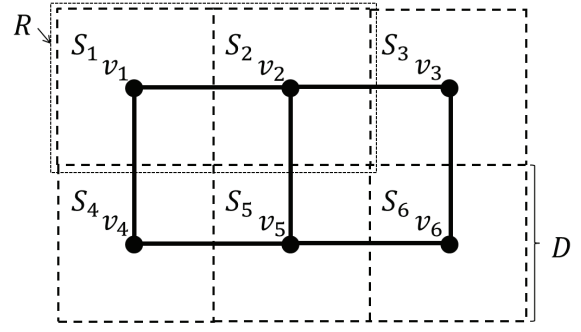}
\caption{Square lattice $L_D$ formed by the bonds $(v_i,v_j)$.}
\label{onedependent}
\end{centering}
\end{figure}

\begin{lem}
Let the state of a bond $(v_1,v_2)$ be determined by the homogeneous Poisson point processes of $G'$ inside $R$, and the conditions for a bond to be open be identical for all bonds. Then the bonds form a 1-dependent bond percolation model $L_D$ with identical bond open probabilities.
\end{lem}
\iftechreport{
\begin{proof}
The event that a bond is open is not independent of the event that its adjacent bond is open, since the two events both depend on the point process in an overlapping square. However, the event that a bond is open is independent of the event that any non-adjacent bond is open, since their associated rectangles do not overlap and the point processes in the two rectangles are independent.

Moreover, a Poisson point process is invariant under translation and rotation. Given that the points in $G'$ are generated by homogeneous Poisson point processes and the conditions for a bond to be open are identical, the probability that a bond is open is identical for all bonds.
\end{proof}
}
{The proof of this lemma can be found in the technical report \cite{reportRGG}. }By properly setting the conditions for a bond to be open, the
percolation of $L_D$ can imply the percolation of $G'$. We first look at an example in \cite{bollobas2006percolation} that studies the percolation of $G(\lambda,d)$, and then extend the technique to study $\GG$.

{\it Example \cite{bollobas2006percolation}: }
Let a bond be open if a path in $G(\lambda,d)$ crosses\footnote{A path crosses a rectangle $R' = [x_1, x_2] \times [y_1,y_2]$ horizontally if the path consists of a sequence of connected nodes $v_1, v_2, \dots, v_{n-1}, v_n$, and $v_2, \dots, v_{n-1}$ are in $R'$, $x(v_1) \leq x_1, x(v_n) \geq x_2$, $y_1 \leq y(v_1), y(v_n) \leq y_2$, where $x(v_i)$ is the $x$-coordinate of $v_i$ and $y(v_i)$ is the $y$-coordinate of $v_i$. A path crosses a rectangle vertically is defined analogously.} $R'$ horizontally and another path in $G(\lambda,d)$ crosses $S'_1$ vertically, where $R'$ is a $(2D - 2d) \times (D - 2d)$ rectangle that has the same center as $R$, and $S'_1$ is a $(D - 2d) \times (D - 2d)$ square that has the same center as $S_1$. The reason for considering $R'$ and $S'_1$ is that the existence of the two crossing paths over $R'$ and $S'_1$ is determined by the point process within $R$, while the existence of links within distance $d$ from the boundaries (and thus the crossings over $R$) may depend on nodes outside $R$.

If two adjacent bonds are open, the paths in $G(\lambda,d)$ in the two rectangles are joined. To see this, note that in Fig.~\ref{ConnectedCompleteRec}, if the black and blue bonds (same direction) are both open, the crossings 1 and 2 intersect. If the black and red bonds (perpendicular) are both open, the crossings 1 and 3 intersect.

\begin{figure}[h]
\begin{centering}
\leavevmode\includegraphics[width=0.9\linewidth]{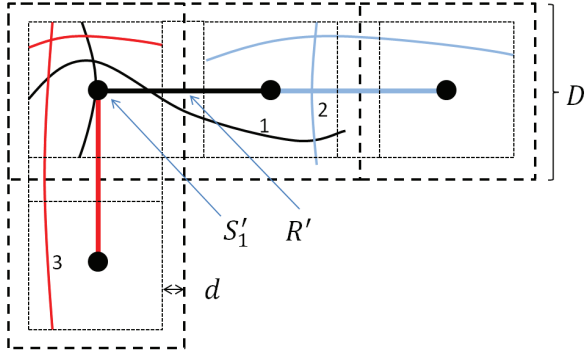}
\caption{Crossings over rectangles associated with two adjacent open bonds are joined.}
\label{ConnectedCompleteRec}
\end{centering}
\end{figure}

If the square lattice $L_D$ percolates, open bonds form an infinite component. Paths in $G(\lambda, d)$ across the rectangles associated with the open bonds are connected and form an infinite component. Therefore, a node density above which $L_D$ percolates is an upper bound on the percolation threshold of $G(\lambda, d)$.

%\subsubsection{Mapping $\GG$ to $L_D$}
\subsubsection{Proof of Theorem \ref{th:upperlarge}}
\label{sc:dep}
We map $\GG$ to $L_D$ by letting a bond in $L_D$ be open if the following three conditions are satisfied in its associated rectangle $R = S_1 \cup S_2$. The size of the rectangle satisfies $D = \min (d_2 / \sqrt{10}, d_\text{dep} / \sqrt 5)  \geq  d_2 / 2\sqrt 5$.
\begin{enumerate}
\item A path from $G_1$ crosses $R'$ horizontally, where $R'$ is a $(2D - 2d_1) \times (D - 2d_1)$ rectangle that has the same center as $R$.
\item A path from $G_1$ crosses $S'_1$ vertically, where $S'_1$ is a $(D - 2d_1) \times (D - 2d_1)$ square that has the same center as $S_1$.
\item There is at least one node from $G_2$ in $R$.
\end{enumerate}

\iftechreport{
To see that the percolation of $L_D$ implies the percolation of $\GG$, consider any two adjacent open bonds in $L_D$. In the two rectangles associated with the bonds,
1)
paths from $G_1$ that cross one rectangle are joined with paths from $G_1$ that cross the other rectangle;
2)
at least two nodes from $G_2$, one in each rectangle, are connected by a link in $G_2$, because any two nodes in adjacent rectangles are within distance $\sqrt {10} D \leq d_2$;
3)
every node in $G_i$ has at least one supply node in $G_j$ inside the rectangle ($\forall i,j \in \{1,2\}, i \neq j$), in which the distance between two nodes is no larger than $\sqrt 5 D \leq d_\text{dep}$.
}{To bound the percolation thresholds of $\GG$, we prove in \cite{reportRGG} that the percolation of $L_D$ implies the percolation of $\GG$, and compute the probability that the three conditions are satisfied using a method similar in \cite{Tse}.}

\iftechreport{
If the probability $p_{123}$ that a bond is open is above 0.8639, then $L_D$ percolates and $\GG$ also percolates. An upper bound on the percolation threshold of $\GG$ is a pair of node densities $(\lambda_1,\lambda_2)$ that yields $p_{123} \geq 0.8639$.
In the remainder of the proof, we compute $p_{123}$ as a function of $(\lambda_1,\lambda_2)$.

To determine the probability that the first and the second conditions are satisfied, we consider a discrete square lattice represented by Fig.~\ref{latticeCrossing}. Bonds of length $d_1/2$ form a square lattice $L'$ in a {\it finite} $md_1 \times md_1/2$ region, where $m = \lfloor 2D/d_1 \rfloor$. Let a bond in $L'$ be open if there is at least one node from $G_1$ in the $d_1/2\sqrt 2 \times d_1/2\sqrt 2$ square that contains the bond (the small square that has dashed boundaries in the figure), which occurs with probability $ p = 1 - e^{-\lambda_1 d_1^2/8}$.
It is clear that if the open bonds form a horizontal crossing\footnote{A horizontal crossing of open bonds over a rectangle $R'=[x_1,x_2] \times [y_1, y_2]$ consists of a sequence of adjacent open bonds in the rectangle such that at least one bond has an endpoint with $x$-coordinate $x_1$ and at least one bond has an endpoint with $x$-coordinate $x_2$. A vertical crossing of open bonds is defined analogously.} over $L'$, then nodes in $G_1$ form a horizontal crossing path over $R'$. %In $L'$, the probability that a bond is open is
%$ p = 1 - e^{-\lambda_1 d_1^2/8}. $ %If the open bonds form a vertical crossing over the left half of $L'$, then nodes of $G_1$ form a vertical crossing over $S'$.
\begin{figure}[h]
\begin{centering}
\leavevmode\includegraphics[width=.95\linewidth]{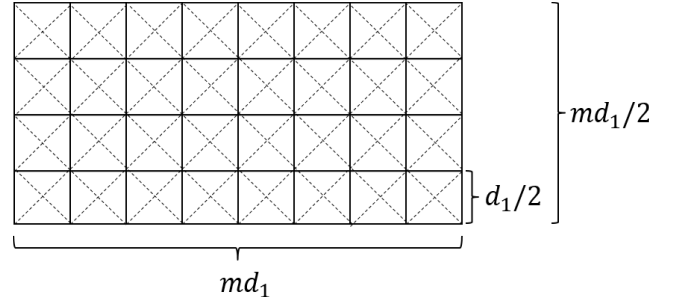}
\caption{Mapping the crossing in $G_1$ to the crossing in a square lattice $L'$.}
\label{latticeCrossing}
\end{centering}
\end{figure}

Let $p_x(km, m, p)$ denote the probability that there exists a horizontal crossing over the $km \times m$ square lattice $L'$ given that each bond is open independently with probability $p$. A lower bound on $p_x(km, m, p)$, Eq. (\ref{eq:crossing}), can be derived by a standard technique in percolation theory (\eg, an extension of Proposition 2 in \cite{Tse}).
\begin{equation}\label{eq:crossing}
  p_x(km, m, p) \geq 1 - \frac{4}{3}(km+1) e^{m \log 3(1-p)}.
\end{equation}
The probability that the crossing exists is close to 1 if $m$ is large and $p > 2/3$.

Finally, the probability that the first condition is satisfied is $p_1 \geq p_x(2m, m, p)$. The probability that the second condition is satisfied is $p_2 \geq p_x(m, m, p)$. Given that the existence of the two crossings are positively correlated, by the FKG inequality \cite{bollobas2006percolation}, the probability that both conditions are satisfied is lower bounded by:
$$p_{12} \geq p_1p_2 \geq p_x(2m, m, p)p_x(m, m, p).$$

The probability that there is at least one node from $G_2$ in $R$ (\ie, the third condition is satisfied) is
$p_3 = 1 - e^{-2 D^2 \lambda_2 }$.
Given that the point processes in $G_1$ and $G_2$ are independent, the probability that a bond is open is $p_{123} = p_{12} p_3$. As long as $p_{123} > 0.8639$, $\GG$ percolates. This completes the proof.
}{}

\subsubsection{An example of two RGGs with large $d_2/d_1$}
%We apply Theorem \ref{th:upperlarge} to study the percolation of two interdependent RGGs that have a large connection distance ratio $d_2/d_1$. The result implies that if the ratio $d_2/d_1$ is very large, the increase of $\lambda_i$ above certain threshold has little impact on the critical $\lambda_j$ for $\GG$ to percolate ($\forall i,j\in\{1,2\},i \neq j$).

%{\it Example: }
We study two interdependent RGGs $G_1$ and $G_2$, which have a finite number of nodes, in order to quantify $d_2/d_1$ as a function of the number of nodes in the graph.
If $d_2 = \Omega (d_1 \log n_1)$, and $d_\text{dep} \geq d_2 /2$, then $m = \Omega(\log n_1)$, where $n_1$ is the expected number of nodes in $G_1$. As $n_1$ approaches infinity, the probability $p_x(km, m, p)$ approaches 1 if $p > 2/3$, by Eq.~\ref{eq:crossing}. %Both $m$ and $p_x(km, m, p)$ have been defined in the proof of Theorem \ref{th:upperlarge}.

Applying Theorem \ref{th:upperlarge}, by solving
$p = (1 - e^{-\lambda_1 d_1^2/8}) = 2/3$, and
$p_3 = 1 - e^{-2 D^2 \lambda_2 } = 0.8639,$ we obtain an upper bound on percolation threshold $\lambda_1 = 8.789/d_1^2$, $\lambda_2 = 19.94/d_2^2$.
The bounds suggest that if the ratio between the connection distances of two RGGs is very large, the node density in one RGG may {\it not} affect the minimum node density in the other RGG in order for the giant mutual component to exist in the interdependent RGGs. %The increase of the critical node densities due to the interdependence between two RGGs becomes less significant as the difference between their scales becomes large. In fact, it is tempting to

We conjecture that as long as the node density of each individual RGG is above the percolation threshold of the single graph, then the interdependent RGGs percolate, if $d_1 \ll d_2$ and $\dd = (1 + \epsilon)d_2/2$ for $\epsilon > 0$. This can be intuitively explained as below. Let $V^0_2$ denote the nodes in the giant component of a single graph $G_2$ without considering the interdependence. Disks of radius $d_2/2$ centered at nodes in $V^0_2$ are connected. Disks of radius $\dd > d_2/2$ centered at nodes in $V^0_2$ are also connected, and this region contains nodes in $G_1$ that have functional supply nodes. Each disk of radius $\dd$ is so large compared with $d_1$, that the probability that there is a crossing formed by connected nodes in $G_1$ along any direction across the disk approaches one\footnote{If nodes are generated by a Poisson point process with density above the percolation threshold, the probability that there is a horizontal path across a $kl \times l$ rectangle approaches one for any $k$ as $l \rightarrow \infty$ \cite{meester1996continuum}.}. Moreover, the disks of radius $d_\text{dep}$ have overlaps with width and height at least $\epsilon d_2 \gg d_1$, which are sufficiently large to join the paths in $G_1$ across two overlapping disks. Thus, a giant component of $G_1$ exists near the giant component of $G_2$. Nodes in the two components are interdependent and form a giant mutual component.

\subsection{Numerical results}
We verify the bounds in Theorem \ref{th:uppersmall} by simulating $\GG$ in a $10 \times 10$ square. Table \ref{tb:small} illustrates the fraction of nodes from $G_i$ that belong to the largest mutual component, denoted by $f_i$, ($\forall i \in \{1,2\}$). The fractions are averaged over 5 instances of simulations for each combination of $(\lambda_1,\lambda_2,d_1,d_2,\dd)$ that satisfies the condition in Theorem \ref{th:uppersmall}. To verify the bounds in Theorem \ref{th:upperlarge}, we simulate $\GG$ in a $30 \times 30$ square (to simulate a sufficiently large $G_2$ under small node densities). Table \ref{tb:large} illustrates the average fraction of nodes in the largest mutual component, for $(\lambda_1,\lambda_2,d_1,d_2,\dd)$ given by Theorem \ref{th:upperlarge}. We observe that most nodes in $G_1$ and $G_2$ belong to the largest mutual component, which implies that $\GG$ percolates.

\begin{table}[h]
\centering
\caption{Fraction of nodes in the largest mutual component under the condition of Theorem \ref{th:uppersmall}}
\label{tb:small}
\begin{tabular}{|l|l|l|l|l|l|l|}
\hline
$\lambda_1$ & $\lambda_2$ & $d_1$ & $d_2$ & $\dd$  & $f_1$   & $f_2$ \\ \hline
15      & 1.54    & 1  & 3  & 1.5 & 1.00    & 1.00  \\ \hline
20      & 0.92    & 1  & 3  & 1.5 & 0.99 & 1.00  \\ \hline
25      & 0.75    & 1  & 3  & 1.5 & 0.98 & 1.00  \\ \hline
15      & 2.39    & 1  & 2  & 1   & 0.99 & 1.00  \\ \hline
20      & 1.80     & 1  & 2  & 1   & 1.00    & 1.00  \\ \hline
25      & 1.58    & 1  & 2  & 1   & 0.97 & 1.00  \\ \hline
\end{tabular}
\end{table}

\begin{table}[h]
\centering
\caption{Fraction of nodes in the largest mutual component under the condition of Theorem \ref{th:upperlarge}}
\label{tb:large}
\begin{tabular}{|l|l|l|l|l|l|l|}
\hline
$\lambda_1$ & $\lambda_2$ & $d_1$ & $d_2$ & $\dd$   & $f_1$   & $f_2$ \\ \hline
16      & 0.190    & 1  & 10 & 7.07 & 1.00    & 1.00  \\ \hline
17      & 0.123   & 1  & 10 & 7.07 & 1.00    & 1.00  \\ \hline
25      & 0.100     & 1  & 10 & 7.07 & 1.00    & 1.00  \\ \hline
17      & 0.385  & 1  & 8  & 5.66 & 1.00    & 1.00  \\ \hline
18      & 0.207   & 1  & 8  & 5.66 & 1.00    & 1.00  \\ \hline
25      & 0.156   & 1  & 8  & 5.66 & 0.99 & 1.00  \\ \hline
\end{tabular}
\end{table}

\color{black}
{\it Remark: }We have assumed that $d_\text{dep} \geq \text{max}(d_1/2, d_2/2) = d_2/2$ throughout this section. To see that this is a reasonable assumption, note that nodes in $G_1$ that have at least one functional supply node are restricted in the region $R_\text{dep}$, where $R_\text{dep}$ is the union of disks with radius $d_\text{dep}$ centered at nodes in the infinite component of $G_2$. If $R_\text{dep}$ is fragmented, it is not likely for disks of radius $d_1/2<d_2/2$ centered at random locations within $R_\text{dep}$ to overlap, and it is not likely that a functional infinite component will exist in $G_1$, unless the node density in $G_1$ is large. Therefore, the interdependent distance $d_\text{dep}$ should be large enough so that $R_\text{dep}$ is a connected region, to avoid a large minimum node density in $G_1$. The region $R_\text{dep}$ can be made larger by increasing either $\lambda_2$ or $d_\text{dep}$. Setting $d_\text{dep} \geq d_2 / 2$ avoids increasing $\lambda_2$ high above the percolation threshold of $G_2$, in order for $R_\text{dep}$ to be connected. In Section \ref{sc:interval}, we develop a more general approach that does not require this assumption.
\section{Confidence intervals for percolation thresholds}
\label{sc:interval}
\textcolor{black}{
In this section, we compute confidence intervals for percolation thresholds. The confidence intervals provide interval estimates for the percolation thresholds. If the node densities in $\GG$ are below the lower confidence bounds, then there does not exist an infinite mutual component in $\GG$ with high confidence. On the other hand, if the node densities are above the upper confidence bounds, then there exists an infinite mutual component in $\GG$ with high confidence. Compared with the analytical upper bounds in Section \ref{sc:upperbound}, the numerical upper confidence bounds are much tighter. Moreover, the techniques in this section apply to $\GG$ with general $d_1, d_2, \dd$.}

The mapping to compute confidence intervals is related to the mapping from $\GG$ to the 1-dependent bond percolation model $L_D$ in Section \ref{sc:onedeplarge}. Both mappings satisfy the following properties: 1) the percolation of $L_D$ implies the percolation of $\GG$; 2) the event that determines the state of a bond depends only on the point process within its associated rectangle, thus preserving the 1-dependent property. The probability that the event occurs can be computed or bounded analytically in the previous section. In contrast, in this section, we consider events whose probabilities are larger under the same point processes but can only be evaluated by simulation. Since the events that we consider in this section are more likely to occur under the same point processes, the mappings yield tighter bounds.

Our mappings from $\GG$ to $L_D$ extend the mappings from $G(\lambda,d)$ to $L_D$ proposed in \cite{balister2005}. For completeness, we first briefly summarize the mappings in \cite{balister2005} that compute upper and lower bounds on the percolation threshold of $G(\lambda,d)$. %Then we define mappings that can be used to obtain upper and lower bounds on the percolation thresholds of $\GG$.

{\it Upper bound for $G(\lambda,d)$ \cite{balister2005}:} Recall Fig.~\ref{onedependent}. The event that a bond $(v_1,v_2) \in L_D$ is open is determined by the point process of $G(\lambda,d)$ in the rectangle $R = S_1 \cup S_2$, where $S_1$ and $S_2$ are squares. Let $V_i$ denote the largest component formed by the points of $G(\lambda,d)$ in $S_i$. If $V_i$ is the {\it unique} largest component in $S_i$ ($\forall i \in \{1,2\}$) and $V_1$ and $V_2$ are connected, then the bond is open. Otherwise, the bond is closed. %See Figs. \ref{connect2} and \ref{disconnect2} for an illustration, where $C_i$ are represented by scattered points. The bond associated with the rectangle is open in Fig. \ref{connect2} while it is closed in Fig. \ref{disconnect2}.

If $L_D$ percolates, open bonds form an infinite component. As a result, the largest components in the squares that intersect the open bonds are connected in $G(\lambda,d)$ and they form an infinite component. Therefore, a node density $\lambda$, above which the probability that %the largest components in adjacent squares are connected
a bond is open is larger than 0.8639, is an upper bound on the percolation threshold of $G(\lambda,d)$.

{\it Lower bound for $G(\lambda,d)$ \cite{balister2005}:} Let the {\it connection process} of $G(\lambda, d)$ be the union of nodes and links in $G(\lambda, d)$. Let the {\it complement} of the connection process be the union of the empty space that does not intersect nodes or links. If the complement of the connection process form a connected infinite region, then all the connected components in $G(\lambda,d)$ have finite sizes and $G(\lambda, d)$ does not percolate%\footnote{The relationship between the connection process and its complement can be equivalently stated by the (more commonly stated) relationship between {\it covered region} and {\it vacant region} of $G(\lambda,d)$, where the covered region is the union of disks of radius $d/2$ centered at all nodes and the vacant region is the complement of the covered region. Below the percolation threshold, there exists a vacant region that has infinite size while all the covered regions have finite sizes. Above the percolation threshold, all the vacant regions have finite sizes \cite{roy1990russo}.} l
~\cite{balister2005, roy1990russo}. Consider the complement of the connection process in rectangle $R$. Let a bond (in $L_D$) associated with rectangle $R$ be open if the complement process forms a horizontal crossing\footnote{The complement of a connection process forms a horizontal crossing over a rectangle if a curve in the rectangle touches the left and right boundaries of the rectangle and the curve does not intersect any nodes or links. The vertical crossing of the complement process is defined analogously.} over the rectangle $R'$ and a vertical crossing over the square $S'_1$. \iftechreport{Recall that rectangle $R'$ is the $(2D - 2d) \times (D - 2d)$ rectangle that has the same center as $R$, and square $S'_1$ is the $(D - 2d) \times (D - 2d)$ square that has the same center as $S_1$, the left square in $R$. For example, in Fig. \ref{complete}, the two crossings that do not intersect any nodes or links are plotted.

\begin{figure}[h]
\begin{centering}
\leavevmode\includegraphics[width=0.95\linewidth]{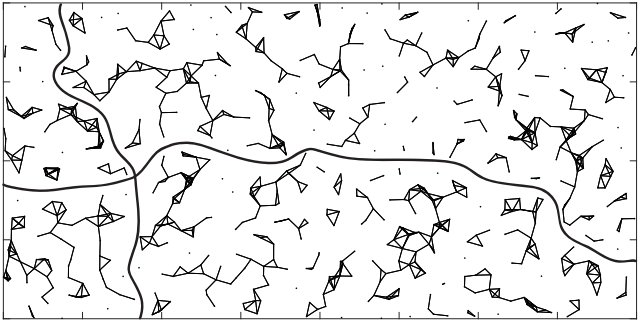}
\caption{The horizontal and vertical crossings from the complement of the connection process over the rectangle.}
\label{complete}
\end{centering}
\end{figure}
}{}

If $L_D$ percolates, the complement process forms an infinite region and $G(\lambda,d)$ does not percolate. To conclude, a node density, under which the probability that the complement process forms the two crossings is above 0.8639, is a lower bound on the percolation threshold for $G(\lambda,d)$.

\subsection{Upper bounds for $\GG$}
\label{sc:upper}

In $G(\lambda,d)$, the largest connected component that contains a node $b$ can be computed efficiently by contracting the links (or using a breadth-first-search) starting from $b$. Two components are connected and form one component if there exists two nodes within distance $d$, one in each component. We next extend these notions to $\GG$. %By analogy, we need an algorithm to compute a mutual component and to characterize conditions for two mutual components to be connected in $\GG$. %In the following we propose an algorithm that obtains a mutual connected component of $\GG$ that contains two given nodes, one from each graph.

Let $G_1$ and $G_2$ denote the two graphs in $\GG$. Let $b_1 \in G_1$ and $b_2 \in G_2$ denote two nodes within the interdependent distance $\dd$. Algorithm \ref{alg:mutual1} computes the largest mutual component $M(b_1, b_2)$ that contains $b_1$ and $b_2$. The correctness follows from the definition of mutual component.
\begin{algorithm}[h]
\caption{Computing the largest mutual component that contains two specified nodes $b_i \in G_i$ within $\dd$ ($\forall i \in \{1,2\}$).}
\begin{enumerate}
\item Find all the nodes $V^0_i(b_i)$ that are connected to $b_i$ (either directly or through a sequence of links) in $G_i$ ($\forall i \in \{1,2\}$).
\item Remove nodes in $V^0_i(b_i)$ that do not have any supply nodes in $V^0_j(b_j)$ ($\forall i,j \in \{1,2\}, i \neq j$). Among the remaining nodes, find the nodes $V^1_i(b_i) \subseteq V^0_i(b_i)$ that are connected to $b_i$ ($\forall i \in \{1,2\}$).
\item Repeat step 2 until $V^{k+1}_i(b_i) = V^k_i(b_i)$ ($\forall i \in \{1,2\}$). Let $M(b_1, b_2) = V^k_1(b_1) \cup V^k_2(b_2)$.
\end{enumerate}
\label{alg:mutual1}
\end{algorithm}
\color{black}

Two mutual components $M = V_1 \cup V_2$ and $\hat M = \hat V_1 \cup \hat V_2$ form one mutual component if and only if $V_i$ and $\hat V_i$ are connected in $G_i$ ($\forall i \in \{1,2\}$). The necessity of the condition is obvious. To see that this condition is sufficient, note that every node in the connected component formed by $V_i$ and $\hat V_i$ has at least one supply node that belongs to the connected component formed by $V_j$ and $\hat V_j$ ($\forall i, j \in \{1,2\}, i \neq j$). The condition can be generalized naturally for more than two mutual components to form one mutual component.

The method of obtaining an upper bound on the percolation threshold of $G(\lambda,d)$ can be modified to obtain an upper bound on the percolation threshold of $\GG$, by declaring a bond to be open if the unique largest mutual components in the two adjacent $D \times D$ squares $S_1$ and $S_2$ are connected. However, computing the largest mutual component of $\GG$ in $S_i$ is not as straightforward as computing the largest component of $G(\lambda,d)$ in $S_i$.
In $G(\lambda,d)$, a node belongs to exactly one (maximal) connected component. All the components can be obtained by contracting the links, and the largest component can be obtained by comparing the sizes of the components. However, in $\GG$, a node may belong to multiple mutual components. For example, let $b_1$ and $b_2$ be two isolated nodes in $G_1$, and let $b_3$ and $b_4$ be two connected nodes in $G_2$. If both $b_1$ and $b_2$ are within the interdependent distance from $b_3$ and $b_4$, $\{b_1, b_3, b_4\}$ and $\{b_2, b_3, b_4\}$ are two mutual components.
An algorithm that computes the largest mutual component of $\GG$ in a square 1) selects a pair of nodes, one from each graph, and computes the largest mutual component that contains the two nodes by Algorithm \ref{alg:mutual1}, and then 2) chooses the largest mutual component over all pairs of nodes in the square within the interdependent distance. Thus, it requires much more computation than finding the largest component of $G(\lambda,d)$ in a square.

Instead of optimizing the algorithm and obtaining the largest mutual component in square $S$, a mutual component $M^\text{greedy}(S)$ can be computed by Algorithm \ref{alg:mutual2}. This algorithm has good performance in finding a large mutual component when the square size is large. In particular, if the square had infinite size, this algorithm would find an infinite mutual component if one exists.

\begin{algorithm}[h]
\caption{An algorithm that greedily computes a mutual component $M^\text{greedy}(S)$ in region $S$.}
\begin{enumerate}
\item Find the largest connected component $V^0_i(S)$ in $G_i(S)$, where $G_i(S)$ consists of the nodes and links of $G_i$ in region $S$. If there is more than one largest connected component, apply any deterministic tie-breaking rule (\eg, choose the component that contains a nodes with the smallest $x$-coordinate).
\item Remove nodes in $V^0_i(S)$ that do not have supply nodes in $V^0_j(S)$ ($\forall i,j \in \{1,2\}, i \neq j$). Find the largest connected component $V^1_i(S)$ formed by the remaining nodes in $V^0_i(S)$ ($\forall i \in \{1,2\}$), and apply the same tie-breaking rule.
\item Repeat step 2 until $V^{k+1}_i(S) = V^k_i(S)$ ($\forall i \in \{1,2\}$). Let $M^\text{greedy}(S) = V^k_1(S) \cup V^k_2(S)$.
\end{enumerate}
\label{alg:mutual2}
\end{algorithm}

Let a bond $(v_1,v_2)$ in $L_D$ be open if the two components $M^\text{greedy}(S_1)$ and $M^\text{greedy}(S_2)$ form one mutual component\iftechreport{}{, where $M^\text{greedy}(S_i)$ is computed by Algorithm \ref{alg:mutual2}. See \cite{reportRGG} for the rationale behind this algorithm}. Since $M^\text{greedy}(S_i)$ is unique in any square $S_i$, a connected component in $L_D$ implies that $\{M^\text{greedy}(S_i)\}$ form one mutual component in $\GG$, where $S_i$ are the squares that intersect the open bonds in the connected component in $L_D$. If the probability that a bond is open is larger than 0.8639, $L_D$ percolates and $\GG$ also percolates.

An alternative condition for a bond to be open is that nodes in $M^\text{greedy}(R)$ form a horizontal crossing over rectangle $R'$ and a vertical crossing over square $S'_1$ in both graphs (recall Fig. \ref{ConnectedCompleteRec} and the condition for two mutual components to form one mutual component). In order for the existence of the two crossings to only depend on the point processes in $R$, in the definition of the $(2D - 2d) \times (D - 2d)$ rectangle $R'$ and the $(D - 2d) \times (D - 2d)$ square $S'_1$, $d = \max(d_1, d_2) + d_\text{dep}$. %In our numerical tests, this condition can sometimes be satisfied more easily than the previous condition.

An upper bound on the percolation threshold can be obtained by either approach.
The smaller bound obtained by the two approaches is a better upper bound on the percolation threshold for $\GG$.

\subsection{Lower bounds for $\GG$}
\label{sc:lower}
In $\GG$, the connection process consists of nodes and links in mutual components. To avoid the heavy computation of mutual components, we study another model in which the connection process $\tilde P_i$ of $G_i$ in the new model {\it dominates}\footnote{One connection process dominates another if the nodes and links in the first process form a superset of the nodes and links in the second process, for any realization of $G_i$.} the connection process $P_i$ of $G_i$ in $\GG$ ($\forall i \in \{1,2\}$).
As a consequence, the complement of the connection process $\tilde P_i^c$ of $G_i$ in the new model is dominated by $P_i^c$ ($\forall i \in \{1,2\}$). If $\tilde P_i^c$ percolates, then $P_i^c$ percolates and $P_i$ does not percolate (\ie, all the components in $P_i$ have finite sizes). If either $P_1$ or $P_2$ does not percolate, then $\GG$ does not percolate. Thus, node densities under which at least one of $\tilde P_1^c$ and $\tilde P_2^c$ percolates are lower bounds on the percolation thresholds of $\GG$.

%Let $N_i$ denote the nodes generated by Poisson point process in RGG $i$.
The new model can be viewed to have a {\it relaxed} supply requirement. In this model, every node (as opposed to nodes in the same mutual component) is viewed as a valid supply node for nodes in the other graph. A node $b_i$ in $G_i$ is removed if and only if there is no node in $G_j$ within the interdependent distance $d_\text{dep}$ from $b_i$ ($\forall i, j \in \{1,2\}, i \neq j$). After all such nodes are removed, the remaining nodes in $G_i$ are connected if their distances are within the connection distance $d_i$. The computation of the connection process $\tilde P_i$ is efficient and avoids the computation of mutual components in $\GG$ through multiple iterations.
%To compare with the original interdependent RGGs model, the number of nodes that have supply nodes is larger than the number of nodes that have supply nodes from a connected component.

The connection process $\tilde P_i$ in the new model dominates $P_i$ in the original model $\GG$. On the one hand, for any realization, all the links in $P_i$ are present in $\tilde P_i$, because all the nodes in a mutual component have supply nodes, and links between these nodes are present in the new model as well. On the other hand, in the new model, nodes in a connected component $\tilde V_i$ in $G_i$ may depend on nodes in multiple components in $G_j$. In contrast, in $\GG$, the nodes in $\tilde V_i$ may be divided into several mutual components, and links do not exist between two disjoint mutual components.

An algorithm that computes a lower bound on the percolation threshold of $\GG$ is as follows. First, compute the connection process $\tilde P_i$ in the new model. Next, in the $2D \times D$ rectangle $R$, consider the complement of the connection process $\tilde P^c_i$. %Let $R'$ denote the $(2D - 2d) \times (D - 2d)$ rectangle that has the same center as the $2D \times D$ rectangle $R$, and let $S'$ denote the $(D - 2d) \times (D - 2d)$ square that has the same center as the left square in $R$.
Let $p_i$ denote the probability that there is a horizontal crossing over $R'$ and a vertical crossing over $S'_1$ in the complement process $\tilde P^c_i$, where $R'$ and $S'_1$ are the same as before. %If $p_i \geq 0.8639$, $\tilde P^c_i$ form an infinite component. Either $\tilde P^c_1$ or $\tilde P^c_2$ forms an infinite component is a sufficient condition for $\GG$ not to percolate.
A lower bound on the percolation threshold of $\GG$ is given by node densities under which $\max(p_1, p_2) \geq 0.8639$.

\subsection{Confidence intervals}
The probability that a bond is open can be represented by an integral that depends on the point processes in the rectangle $R$. However, direct calculation of the integral is intractable; so instead the integral is evaluated by simulation.
In every trial of the simulation, nodes in $G_1$ and $G_2$ are randomly generated by the Poisson point processes with densities $\lambda_1$ and $\lambda_2$, respectively. The events that a bond is open are independent in different trials. Let the probability that a bond is open be $p$ given $(\lambda_1, \lambda_2)$. The probability that a bond is closed in $k$ out of $N$ trials follows a binomial distribution. The interval $[0.8639, 1]$ is a $99.5 \%$ confidence interval \cite{casella2002statistical} for $p$, given that $N = 100$ and $k = 5$. If $k < 5$, $p \in [0.8639,1]$ with a higher confidence. This suggests that if $k \leq 5$, with $99.5\%$ confidence, $ p \geq 0.8639$ and the 1-dependent bond percolation model $L_D$ percolates given $(\lambda_1,\lambda_2)$.

Based on this method, with $99.5\%$ confidence, an upper bound on the percolation threshold of $\GG$ can be obtained by declaring a bond to be open using the method in Section \ref{sc:upper}, and a lower bound can be obtained by declaring a bond to be open using the method in Section \ref{sc:lower}. For a fixed $\lambda^*_2$, a $99\%$ confidence interval for $\lambda^*_1$ is given by the interval between the upper and lower bounds. Confidence intervals for different percolation thresholds can be obtained by changing the value of $\lambda^*_2$ and repeating the computation. We make a similar remark as in \cite{balister2005}. The confidence intervals are rigorous, and the only uncertainty is caused by the stochastic point processes in the $2D \times D$ rectangle. This is in contrast with the confidence intervals obtained by estimating whether $\GG$ percolates based on extrapolating the observations of simulations in a finite region (which is usually not very large because of limited computational power).

\subsection{Numerical results}
The simulation-based confidence intervals are much tighter than the analytical bounds.
Given that $d_1 = d_2 = 2d_{\text{dep}} = 1$, and $\lambda^*_2 = 2$, the upper and lower bounds on $\lambda^*_1$ are 2.25 and 1.80, respectively, both with $99.5\%$ confidence. In contrast, even if $\lambda^*_2 \rightarrow \infty$, the analytical upper bound on $\lambda^*_1$ is no less than 3.372, which is the best available analytical upper bound for a single $G_1$ \cite{hall1985continuum}. Confidence intervals for the percolation thresholds are plotted in Fig. \ref{ci1}, where the intervals between bars are $99\%$ confidence intervals. %The rectangle sizes are up to $300 \times 150$ in our simulation. Smaller confidence intervals can be obtained by increasing the size of the rectangle. %In particular, as the size approaches infinity, the upper bound approaches the percolation threshold, because the probability that the two crossings from the connection process of $\GG$ exist over the rectangle approaches 1 if $\GG$ percolates.
%However, even if the rectangle has large size, the lower bound may not approach the percolation threshold, because, to calculate the lower bound, we studied a model which may percolate when $\GG$ does not. Obtaining better mappings to lower bound the percolation threshold is left as future work.
\begin{figure}[h]
\begin{centering}
\leavevmode\includegraphics[width=\linewidth]{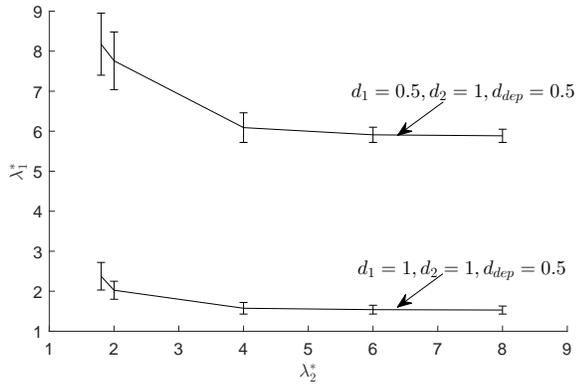}
\caption{The $99\%$ confidence intervals for percolation thresholds of $\GG$ with different connection distances.}
\label{ci1}
\end{centering}
\end{figure}

\textcolor{black}{
To verify the confidence intervals, we simulate $\GG$ within a $20 \times 20$ square, for $d_1 = d_2 = 2 \dd = 1$. Nodes in the largest mutual component are colored black, while the remaining nodes are colored blue. In Fig. \ref{intupper}, the node densities are at the upper confidence bound ($\lambda_1 = 2.25, \lambda_2 = 2.00$), and there exists a mutual component that consists of a large fraction of nodes. In Fig. \ref{intlower}, the node densities are at the lower confidence bound ($\lambda_1 = 1.80, \lambda_2 = 2.00$), and the size of the largest mutual component is small.}

\begin{figure}[h]
\begin{centering}
\leavevmode\includegraphics[width=\linewidth]{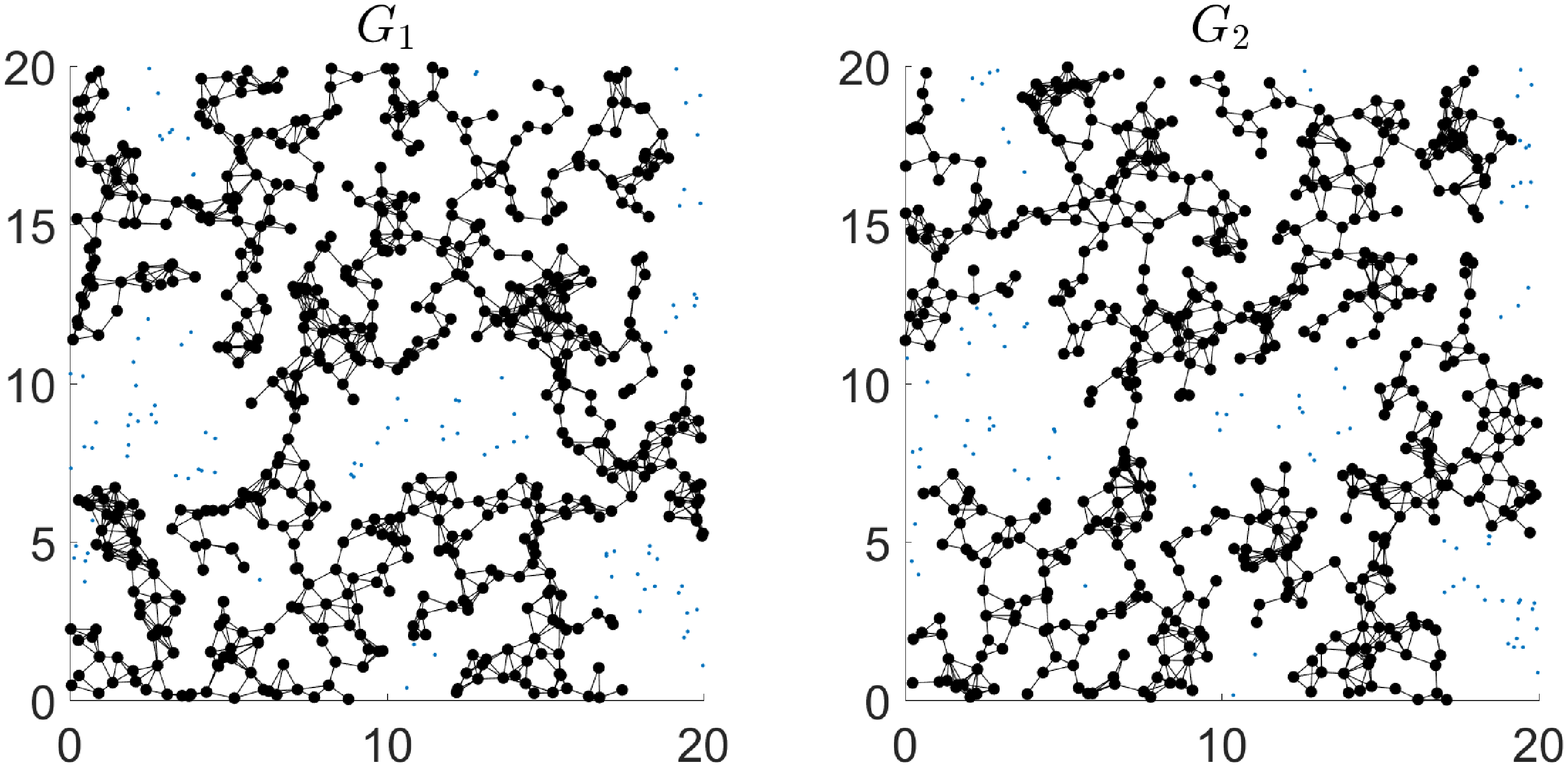}
\caption{The largest mutual component for $\lambda_1 = 2.25, \lambda_2 = 2.00, d_1 = d_2 = 2 \dd = 1$.}
\label{intupper}
\end{centering}
\end{figure}

\begin{figure}[h]
\begin{centering}
\leavevmode\includegraphics[width=\linewidth]{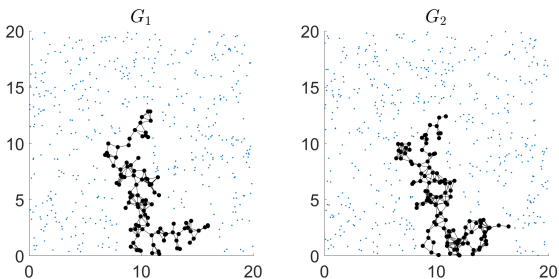}
\caption{The largest mutual component for $\lambda_1 = 1.80, \lambda_2 = 2.00, d_1 = d_2 = 2 \dd = 1$.}
\label{intlower}
\end{centering}
\end{figure}

We next study the impact of interdependent distance $\dd$ on the percolation thresholds. Given $d_1, d_2, \lambda^*_2$, a smaller $\dd$ leads to a higher $\lambda^*_1$, since the probability that a node in $G_1$ has at least one supply nodes from $G_2$ decreases for a smaller $\dd$. The effect is more significant when the number of nodes in $G_2$ is small. This is consistent with Fig. \ref{ci}, where the increase of $\lambda^*_1$ is more significant as $\dd$ decreases when $\lambda^*_2$ is small.

The confidence intervals confirm that the reduced node density due to a lack of supply nodes is not sufficient to characterize the percolation of one of the interdependent graphs. The average density of nodes in $G_1$ that have at least one node within $\dd$ in $G_2$ is $\tilde \lambda_1 = \lambda_1 (1 - e^{-\lambda_2 \pi \dd^2})$, given that $e^{-\lambda_2 \pi \dd^2}$ is the probability that there is no node in $G_2$ within a disk area $\pi \dd^2$. If $\lambda^*_2 = 1.8$, with $99\%$ confidence, $\lambda^*_1 \in [2.03, 2.72]$ when $\dd = 0.5$, and $\lambda^*_1 \in [7.50, 11.20]$ when $\dd = 0.25$. We observe that the ranges of $\tilde \lambda^*_1$ are different: $\tilde \lambda^*_1 \in [1.54, 2.06]$ when $\dd = 0.5$, and $\tilde \lambda^*_1 \in [2.23, 3.33]$ when $\dd = 0.25$. Intuitively, nodes in $G_1$ that have at least one supply nodes are clustered around the nodes in $G_2$, smaller $\dd$ leads to a more clustered point process. The critical node density of a clustered point process is not the same as the critical node density of the homogeneous Poisson point process for percolation. More detailed study on the percolation of a clustered point process can be found in \cite{blaszczyszyn2014comparison}.
\begin{figure}[h]
\begin{centering}
\leavevmode\includegraphics[width=\linewidth]{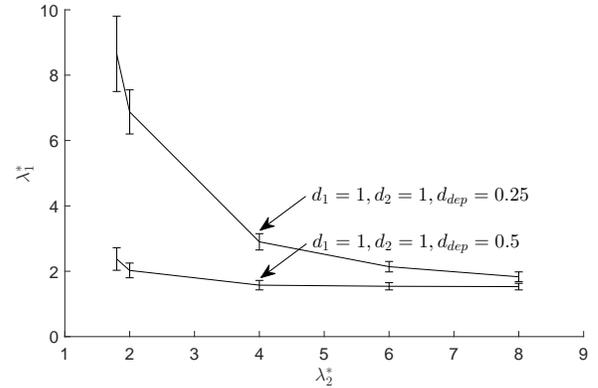}
\caption{The $99\%$ Confidence intervals for percolation thresholds of two $\GG$ with different interdependence distances.}
\label{ci}
\end{centering}
\end{figure}

%We observe from the numerical results that the larger the difference between $d_1$ and $d_2$, the more ``independent'' of the two graphs. Changing $d_1$ from 1 to 0.5, with $99.5\%$ confidence, the upper and lower bounds on $\lambda_1$ are 8.74 and 6.72, respectively. This is smaller than the bounds on $\lambda_1$ when $d_1 = 1$ after scaling $(8.74/4 = 2.19 < 2.55, 6.72/4 = 1.68)$.

%For $\lambda_2 = 4$, the confidence intervals of $\lambda_1$ are smaller. If $d_1 = 0.5$, the confidence interval is $[4.76, 6.85]$. If $d_1 = 1/3$, the confidence interval is $[11.03, 14.87]$.
\color{black}
\section{Robustness of interdependent RGGs under random and geographical failures}
\label{sc:robustness}
Removing nodes independently at random with the same probability in one graph is equivalent to reducing the node density of the Poisson point process. To study the robustness of $\GG$ under random failures, the first step is to obtain the upper and lower bounds on percolation thresholds. With the bounds, we can determine which graph is able to resist more random node removals, by comparing the gap between the node density $\lambda_i$ and the percolation threshold $\lambda^*_i$ given $\lambda_j$ ($i,j \in \{1,2\}, i \neq j$). The graph that can resist a smaller fraction of node removals is the bottleneck for the robustness of $\GG$. Moreover, we are able to compute the maximum fraction of nodes that can be randomly removed from two graphs while guaranteeing $\GG$ to be percolated.
\begin{comment}
the answers to the following questions are obvious.
\begin{itemize}
\item Which network is more vulnerable to random failures in the interdependent networks? Equivalently, what fraction of nodes in $G_i$ can be randomly removed while the interdependent networks still percolate?
\item What is the maximum fraction of random nodes failures that the interdependent networks can tolerate?
\end{itemize}
\end{comment}

We next show that $\GG$ still percolates after a geographical attack that removes nodes in a finite connected region, if the node densities of the two graphs before the attack are above any {\it upper bound} on the percolation thresholds obtained in this paper (either analytical or simulation-based). Recall that we obtained upper bounds on the percolation thresholds of $\GG$ by mapping the percolation of $\GG$ to either the independent bond percolation on a square lattice $L$ or the 1-dependent bond percolation on a square lattice $L_D$. Under both mappings, the event that a bond $e$ is open is entirely determined by the point processes in a finite region $R_e$ that contains the bond. After removing nodes of $\GG$ in a connected finite geographical region, the state of a bond $e$ may change from open to closed only if $R_e$ intersects the attack region. Let $R_f$ be the union of $R_e$ that intersects the attack region. The region $R_f$ is also a connected finite region. As long as $L$ or $L_D$ still percolates after setting bonds in $R_f$ to be closed, $\GG$ percolates.

Results from the percolation theory indeed indicate that setting all the bonds in a finite region $R_f$ to be closed does not affect the percolation of $L$ or $L_D$. For any percolated $L$, the probability that there exists a horizontal crossing of open bonds over a $kl \times l$ rectangle approaches 1 for any integer $k>1$, as $l \rightarrow \infty$ (Lemma 8 on Page 64 of \cite{bollobas2006percolation}). The percolation of $L$ (after setting all bonds in $R_f$ to be closed) is justified by the fact that the connected open bonds across rectangles form a square annulus that does not intersect $R_f$ (shown in Fig. \ref{geoAnnulus}), which is a standard approach to prove the percolation of $L$ \cite{bollobas2006percolation}. Moreover, the percolation of $L_D$ after all bonds in $R_f$ are set closed can be proved in the same approach, by noting that the probability that open bonds of $L_D$ form a horizontal crossing over a rectangle approaches 1 as the rectangle size increases to infinity \cite{balister2005}.

If the $kl \times l$ rectangle is large but finite, the probability that a horizontal crossing formed by open bonds exists is close to 1 if $L$ or $L_D$ percolates. Therefore, the same analysis demonstrates the robustness of two finite interdependent RGGs under a geographical attack that removes the nodes in a disk region of size $\beta a^2$, where $0 < \beta < 1$.
\iftechreport{}{Simulations for the robustness of interdependent RGGs under geographical attacks can be found in~\cite{reportRGG}.}
\begin{figure}[h]
\begin{centering}
\leavevmode\includegraphics[width=0.5\linewidth]{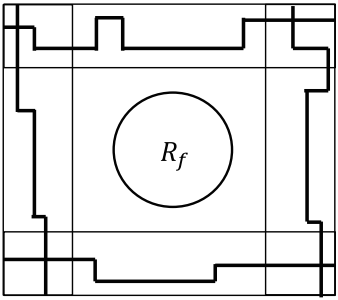}
\caption{Open bonds form a connected path across rectangles around $R_f$.}
\label{geoAnnulus}
\end{centering}
\end{figure}

\iftechreport{
The robustness of interdependent RGGs under geographical failures is illustrated in Fig.~\ref{geo2}. %In Fig. \ref{geo2}, two RGGs have the same node densities and connection distances, while in Fig. \ref{geo3}, two RGGs have different node densities and connection distances.
Nodes and links in the giant mutual component are colored black. The interdependent RGGs still percolate after all the nodes in a disk region are removed. This is in contrast with the cascading failures observed in \cite{berezin2015localized} in the interdependent lattice model after an initial disk attack. One reason may be that every node can have more than one supply node in our model, while every node has only one supply node in \cite{berezin2015localized}. %Another reason is that the node densities of two RGGs are above the upper bounds on the percolation thresholds. A little more redundancy of connection links due to a higher node density and
The multiple localized interdependence helps the interdependent RGGs to resist geographical attacks.
\begin{figure}[h]
\begin{centering}
\includegraphics[width=1.05\linewidth, height = 1.5in]{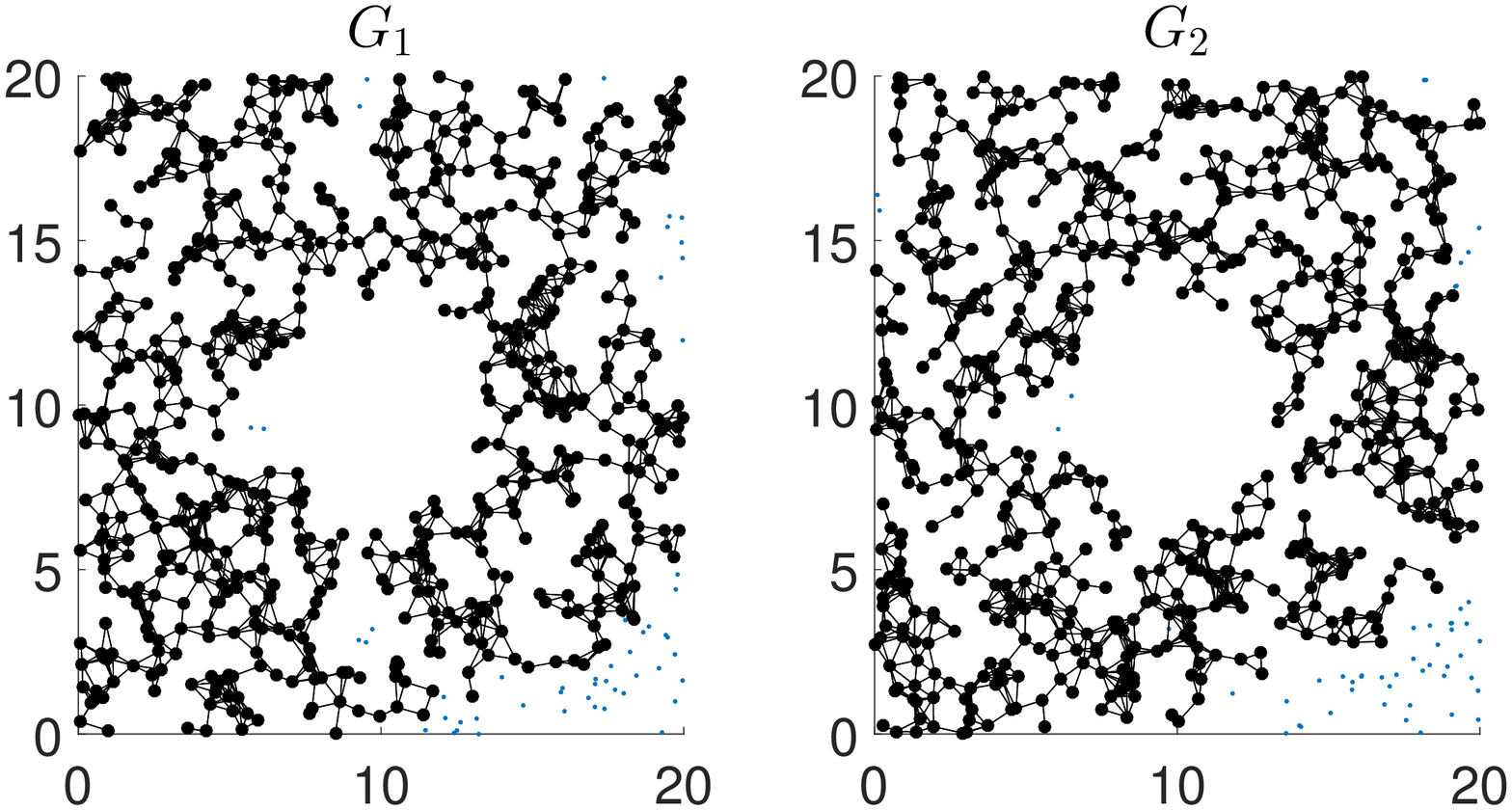}
\caption{Interdependent RGGs with the same connection distance $d_1 = d_2 = 1$ and $\dd = 0.5$.}
\label{geo2}
\end{centering}
\end{figure}
}{}
%\begin{figure}[h]
%\begin{centering}
%\leavevmode\includegraphics[width=\linewidth]{geo3.eps}
%\caption{Interdependent RGGs with different connection distances $d_1 = 0.5, d_2 = 1$.}
%\label{geo3}
%\end{centering}
%\end{figure}

%a single graph $\tilde G$. Even if the nodes in a region are all removed, there exist paths from $\tilde G$ across rectangles that do not intersect the attack region. The percolation of $\tilde G$ implies the percolation of the interdependent graphs. This is in contrast with the interdependent lattice model studied in \cite{berezin2015localized} where every node has exactly one interdependent node and geographical attack can destroy the percolation.

\section{Extensions to more general interdependence}
\label{sc:extension}
In the previous sections, we studied a model where every node in $G_i$ is content to have at least one supply node in $G_j$ in the same mutual component ($\forall i,j \in \{1,2\}, i \neq j$). The techniques can be extended to study models where every node in $G_i$ must have at least $K_j$ supply nodes from $G_j$ to receive enough supply, where $K_j$ can be either a constant or a random variable ($\forall i,j \in \{1,2\}, i \neq j$). We briefly discuss the extensions to models with more general supply requirement using the example in Section \ref{sc:example}, where $d_1 = d_2 = 2 \dd$.

\subsection{Deterministic supply requirement}
The extension is straightforward if $K_i$ is a constant, $\forall i \in \{1,2\}$. By the same discretization technique, the state of a site in the triangle lattice is determined by the point processes in a cell of area $A$ (recall Fig. \ref{lattice}). Declare a site to be open if there are at least $K_i$ nodes from $G_i$ in the cell that contains the site ($\forall i,j \in \{1,2\}, i \neq j$). For each open site, every node from $G_i$ in the cell has at least $K_j$ supply nodes from $G_j$ in the same cell, satisfying the supply requirement. Following the same analysis as that in Section \ref{sc:example}, the percolation of the triangle lattice implies the percolation of $\GG$.

For a Poisson point process of density $\lambda_j$, the probability that there are at least $K_j$ nodes in a cell of area $A$ is $1 - \sum_{l = 0}^{K_j - 1} (\lambda_j A)^l e^{-\lambda_j A} / l!$. An upper bound on the percolation thresholds is given by $(\lambda_1, \lambda_2)$ that satisfies:
$$ \left[1 - \sum_{l = 0}^{K_1 - 1} \frac{(\lambda_1 A)^l e^{-\lambda_1 A}}{l!} \right] \left[1 - \sum_{l = 0}^{K_2 - 1} \frac{(\lambda_2 A)^l e^{-\lambda_2 A}}{l!} \right] = \frac{1}{2}.$$

\subsection{Random supply requirement}
Some extra work is necessary if $K_i$ is a random variable, $\exists i \in \{1,2\}$. For simplicity, we first consider the case where $K_1 \geq 1 $ is a constant and $K_2$ is a discrete random variable with a cumulative distribution function $F_{K_2}(x)$, $x \in \mathbb{N}$. Furthermore, we assume that the number of supply nodes needed by every node in $G_1$ is independent. After the discretization, a site in the triangle lattice is open if the following two conditions are satisfied for at least one integer-valued $k_2 \geq 1$.
\begin{enumerate}
  \item There are exactly $k_2$ nodes from $G_2$ in the cell.
  \item There are at least $K_1$ nodes from $G_1$ in the cell, each of which needs no more than $k_2$ supply nodes.
\end{enumerate}
If both conditions are satisfied, at least $K_1$ nodes from $G_1$ and the $k_2$ nodes from $G_2$ each have enough supply. It is easy to see that the percolation of the triangle lattice still implies the percolation of $\GG$.

Next we compute the probability that the two conditions are satisfied. The probability that there are $k_2$ nodes from $G_2$ in the cell is:
$$\Pr(N_2 = k_2) = (\lambda_2 A)^{k_2} e^{-\lambda_2 A} / k_2 ! .$$
The probability that there are $l$ nodes from $G_1$ in the cell is:
$$\Pr(N_1 = l) = (\lambda_1 A)^{l} e^{-\lambda_1 A} / l ! .$$
The probability that a node in $G_1$ needs no more than $k_2$ supply nodes is $F_{K_2}(k_2)$. Since the number of supply nodes needed by every node in $G_1$ is independent, the probability that at least $K_1$ out of the $l$ nodes in $G_1$ each need no more than $k_2$ supply nodes is:
\begin{align*}\label{eq:Kmin}
  \Pr(K_2^{(K_1)} &\leq k_2 | N_1 = l) = \\
  & \sum_{t = K_1}^{l}{l \choose t} [F_{K_2}(k_2)]^{t} [1 - F_{K_2}(k_2)]^{l - t},
\end{align*}
for $K_1 \leq l$, and $\Pr(K_2^{(K_1)} \leq k_2 | N_1 = l) = 0$ for $K_1 > l$.
By the law of total probability, for a given $k_2$, the probability that there exist at least $K_1$ nodes from $G_1$ in the cell that each need no more than $k_2$ supply nodes is:
\begin{equation*}\label{eq:KminTotal}
  \Pr(K_2^{(K_1)} \leq k_2) = \sum_{l \geq K_1} \Pr(N_1 = l) \Pr(K_2^{(K_1)} \leq k_2 | N_1 = l).
\end{equation*}
Since the events that there are exactly $k_2$ nodes from $G_2$ in the cell are mutually exclusive for distinct values of $k_2$, using the law of total probability again, the probability that both conditions are satisfied is:
\begin{align*}
p_{12} &= \sum_{k_2 \geq 1} \Pr(N_2 = k_2) \Pr(K_2^{(K_1)} \leq k_2).
%&= \sum_{k_1 \geq 1} \left\{ \frac{(\lambda_2 A)^{k_1}e^{-\lambda_2 A}}{k_1 !}
%\sum_{l \geq 1} \frac{(\lambda_1 A)^{l}e^{-\lambda_1 A}}{l!} \left[ 1 - (1 - F_{K_1}(k_1))^l \right] \right\}.
\end{align*}
Any $(\lambda_1, \lambda_2)$ that satisfies $p_{12} \geq 1/2$ is an upper bound on the percolation threshold of $\GG$.

Finally, we consider the case where both $K_1$ and $K_2$ are discrete random variables. Suppose that $N_i$ nodes from $G_i$ are in the cell of area $A$. If there exist integers $k_i^* \leq N_i$, such that at least $k_i^*$ nodes from $G_i$ each need no more than $k_j^*$ supply nodes, then the $k_i^*$ nodes from $G_i$ all have enough supply ($\forall i,j \in \{1,2\}, i \neq j$). However, it is difficult to obtain a clean formula of the probability that $(k_1^*, k_2^*)$ exists (to satisfy the condition). The events that $(k_1^*, k_2^*)$ exists are not mutually exclusive for distinct values of $k_1^*$ and $k_2^*$. While it is possible to compute this probability using the inclusion-exclusion formula, the computation is expensive, since the number of choices of $(k_1^*, k_2^*)$ can be large and each term in the inclusion-exclusion formula requires the computation of order statistics.

A practical approach to estimate the probability that nodes have enough supply is by simulation. In each trial of the simulation, $N_i$ nodes are randomly generated in area $A$, where $N_i$ follows a Poisson distribution of rate $\lambda_i A$ ($\forall i \in \{1,2\}$). Then, each of the $N_i$ nodes is tagged with a realization of the random variable $K_j$, which indicates the number of required supply nodes ($\forall i,j \in \{1,2\}, i \neq j$). Let $I$ indicate whether there exist $(k_1^*, k_2^*)$ such that at least $k_i^*$ nodes among the $N_i$ nodes all have tags no more than $k_j^*$ ($\forall i,j \in \{1,2\}, i \neq j$). The value of $I$ can be computed by Algorithm \ref{alg:supply}.

\begin{algorithm}[h]
\caption{An algorithm that determines whether nodes have enough supply.}
\textbf{Initialization:}\\
\hspace{3mm} Sort the $N_i$ realizations of the random variable $K_j$ in the ascending order. Let $K_j^{(t)}, t = 1,\dots, N_i$ be the sorted list ($\forall i,j \in \{1,2\}, i \neq j$). Let $t_1 = t_2 = 1$.\\
\textbf{Main loop:}
\algsetblock[]{While}{EndWhile} {}{}
    \begin{algorithmic}
\While {$I$ is not determined}
\State $t'_1 \gets K_1^{(t_2)}$, $t'_2 \gets K_2^{(t_1)}$.
\If {$t'_1 \leq t_1$ and $t'_2 \leq t_2$}
    \State $I \gets 1.$
\EndIf
\If {$t'_1 > N_1$ or $t'_2 > N_2$}
    \State $I \gets 0.$
\EndIf
\State $t_1 \gets \max(t_1, t'_1)$, $t_2 \gets \max(t_2, t'_2)$.
\EndWhile
\end{algorithmic}
\label{alg:supply}
\end{algorithm}

We now prove the correctness of Algorithm \ref{alg:supply}. For easy presentation, the $N_i$ nodes are referred to as nodes in $G_i$ ($\forall i \in \{1,2\}$). Initially, among the nodes in $G_i$, the algorithm chooses one node that needs the smallest number of supply nodes. To support this node, at least $t'_j = K_j^{(1)}$ nodes need to be in $G_j$. If $t'_1 \leq 1$ and $t'_2 \leq 1$, one node from $G_1$ and one node from $G_2$ suffice to support each other. Otherwise, if $t'_j > 1$, at least $t'_j$ nodes need to be in $G_j$. The $t'_j$ nodes must be supported by $K_i^{(t'_j)}$ nodes from $G_i$. If $K_i^{(t'_j)}$ is larger than the total number of nodes in $G_i$, then there are not enough supporting nodes in $G_i$ and $I = 0$. If $K_1^{(t'_2)} \leq t'_1$ and $K_2^{(t'_1)} \leq t'_2$, then $t'_1$ nodes from $G_1$ support $t'_2$ nodes from $G_2$, and vise versa. Note that $t'_1$ and $t'_2$ never decrease in the iterations, and at least one of them strictly increases in an iteration where $I$ is not determined. If there exists at least one pair $(k_1^*, k_2^*)$, the algorithm terminates with $I = 1$ at the smallest pair for both coordinates, which can be easily shown by contradiction. If no such pair $(k_1^*, k_2^*)$ exists, the algorithm terminates with $I = 0$.

Given $(\lambda_1, \lambda_2)$, by repeating a sufficiently large number of trials, the probability that $I = 1$ can be estimated within a small multiplicative error with high confidence using Monte Carlo simulation. As long as this probability is at least $1/2$, $\GG$ percolates with high confidence.

\color{black}
\section{Conclusion}
\label{sc:conclusion}
We developed an interdependent RGG model for interdependent spatially embedded networks. We obtained analytical upper bounds and confidence intervals for the percolation thresholds. The percolation thresholds of two interdependent RGGs form a curve, which shows the tradeoff between the two node densities in order for the interdependent RGGs to percolate. The curve can be used to study the robustness of interdependent RGGs to random failures. Moreover, if the node densities are above any upper bound on the percolation thresholds obtained in this paper, then the interdependent RGGs remain percolated after a geographical attack. Finally, we extended the techniques to models with more general interdependence. The study of percolation thresholds in this paper can be used to design robust interdependent networks.

%\section*{Appendix: percolation under large $d_2/d_1$}

\bibliographystyle{IEEEtran} %{abbrv}
\bibliography{rgg}

\begin{IEEEbiography}
[{\includegraphics[width=1in,height=1.25in,clip,keepaspectratio]{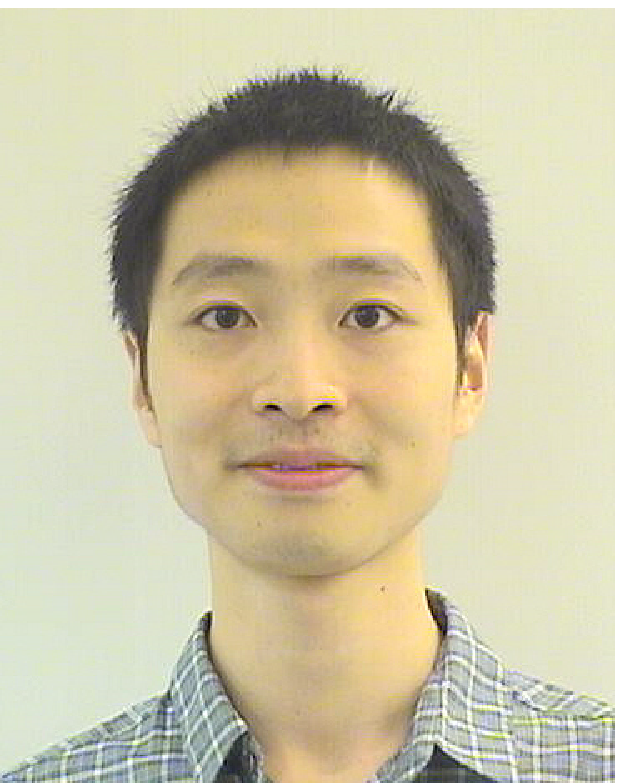}}]
{Jianan Zhang}
received his B.E. degree in Electronic Engineering from Tsinghua University, Beijing, China, in 2012, and M.S. degree from Massachusetts Institute of Technology, Cambridge, MA, USA, in 2014. He is currently pursuing the Ph.D. degree at the Laboratory for Information and Decision Systems, Massachusetts Institute of Technology. His research interests include network robustness, optimization and interdependent networks.
\end{IEEEbiography}

\begin{IEEEbiography}
[{\includegraphics[width=1in,height=1.25in,clip,keepaspectratio]{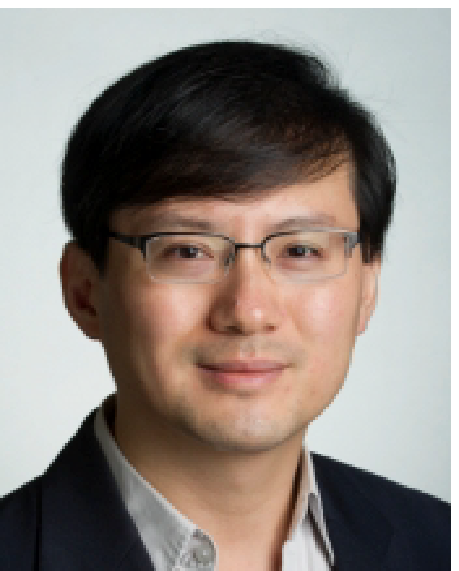}}]
{Edmund M. Yeh}
(SM'12) received his B.S. in Electrical Engineering with Distinction and Phi Beta Kappa from Stanford University in 1994.  He then studied at Cambridge University on the Winston Churchill Scholarship, obtaining his M.Phil in Engineering in 1995.  He received his Ph.D. in Electrical Engineering and Computer Science from MIT in 2001.
He is currently Professor of Electrical and Computer Engineering at Northeastern University.  He was previously Assistant and Associate Professor of Electrical Engineering, Computer Science, and Statistics at Yale University. Professor Yeh has held visiting positions at MIT, Stanford, Princeton, University of California at Berkeley, New York University, Swiss Federal Institute of Technology Lausanne (EPFL), and Technical University of Munich. He has been on the technical staff at the Mathematical Sciences Research Center, Bell Laboratories, Lucent Technologies, Signal Processing Research Department, AT$\&$T Bell Laboratories, and Space and Communications Group, Hughes Electronics Corporation.
Professor Yeh is the recipient of the Alexander von Humboldt Research Fellowship, the Army Research Office Young Investigator Award, the Winston Churchill Scholarship, the National Science Foundation and Office of Naval Research Graduate Fellowships, the Barry M. Goldwater Scholarship, the Frederick Emmons Terman Engineering Scholastic Award, and the President’s Award for Academic Excellence (Stanford University).  He received Best Paper Awards at the ACM Conference on Information Centric Networking, Berlin, September 2017, at the IEEE International Conference on Communications (ICC), London, June 2015, and at the IEEE International Conference on Ubiquitous and Future Networks (ICUFN), Phuket, July 2012.
\end{IEEEbiography}

\begin{IEEEbiography}
[{\includegraphics[width=1in,height=1.25in,clip,keepaspectratio]{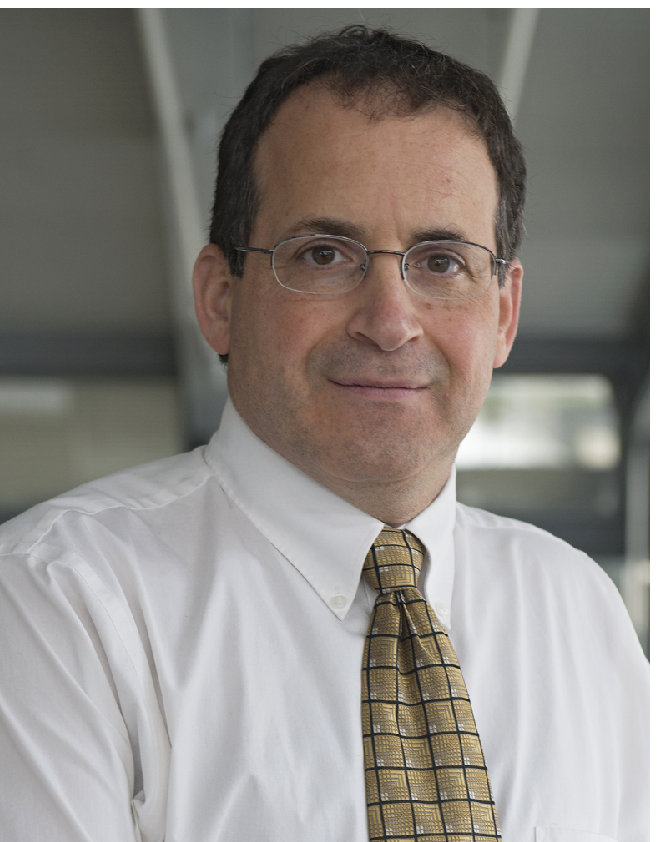}}]
{Eytan Modiano}
received his B.S. degree in Electrical Engineering and Computer Science from the University of Connecticut at Storrs in 1986 and his M.S. and PhD degrees, both in Electrical Engineering, from the University of Maryland, College Park, MD, in 1989 and 1992 respectively.  He was a Naval Research Laboratory Fellow between 1987 and 1992 and a National Research Council Post Doctoral Fellow during 1992-1993.  Between 1993 and 1999 he was with MIT Lincoln Laboratory.  Since 1999 he has been on the faculty at MIT, where he is a Professor and Associate Department Head in the Department of Aeronautics and Astronautics, and Associate Director of the Laboratory for Information and Decision Systems (LIDS).

His research is on communication networks and protocols with emphasis on satellite, wireless, and optical networks.   He is the co-recipient of the MobiHoc 2016 best paper award, the Wiopt 2013 best paper award, and the Sigmetrics 2006 Best paper award.  He is the Editor-in-Chief for IEEE/ACM Transactions on Networking, and served as Associate Editor for IEEE Transactions on Information Theory and IEEE/ACM Transactions on Networking.  He was the Technical Program co-chair for  IEEE Wiopt 2006, IEEE Infocom 2007, ACM MobiHoc 2007, and DRCN 2015.  He is a Fellow of the IEEE and an Associate Fellow of the AIAA, and served on the IEEE Fellows committee.
\end{IEEEbiography}

\end{document}